\documentclass{amsart}
\usepackage[utf8]{inputenc}
\usepackage{amssymb}
\usepackage{amsmath}
\usepackage{amsthm}
\usepackage{amscd}
\usepackage{amstext}
\usepackage{amsfonts}
\usepackage[all]{xy}
\usepackage[dvipdfmx]{graphicx}
\usepackage{hyperref}
\usepackage{bm}

\theoremstyle{plain}
\newtheorem{theorem}{Theorem}[section]
\newtheorem{corollary}[theorem]{Corollary}
\newtheorem{lemma}[theorem]{Lemma}
\newtheorem{proposition}[theorem]{Proposition}
\theoremstyle{definition}
\newtheorem{definition}[theorem]{Definition}
\theoremstyle{remark}
\newtheorem{remark}[theorem]{Remark}
\newtheorem*{example}{Example}
\numberwithin{equation}{section}

\DeclareMathOperator{\ind}{\mathrm{index}}
\DeclareMathOperator{\End}{\mathrm{End}}
\DeclareMathOperator{\Ker}{\mathrm{Ker}}
\DeclareMathOperator{\Coker}{\mathrm{Coker}}
\DeclareMathOperator{\rank}{\mathrm{rank}}
\DeclareMathOperator{\sign}{\mathrm{sign}}

\newcommand{\N}{\mathbb{N}}
\newcommand{\Z}{\mathbb{Z}}
\newcommand{\R}{\mathbb{R}}
\newcommand{\C}{\mathbb{C}}
\newcommand{\T}{\mathbb{T}}
\newcommand{\supp}{\mathrm{supp}}

\newcommand{\I}{\mathcal{I}}
\newcommand{\Corner}{\mathrm{Corner}}
\newcommand{\BE}{\mathrm{BE}}
\newcommand{\id}{\mathrm{id}}
\newcommand{\A}{\mathrm{A}}
\newcommand{\AIII}{\mathrm{A\hspace{-.1em}I\hspace{-.1em}I\hspace{-.1em}I}}
\newcommand{\e}{{\bm e}}
\newcommand{\HH}{\mathcal{H}}
\newcommand{\HHa}{\HH^{\alpha}}
\newcommand{\HHb}{\HH^{\beta}}
\newcommand{\HHab}{\hat{\HH}^{\alpha,\beta}}
\newcommand{\HHs}{\check{\HH}^{\alpha,\beta}}
\newcommand{\TT}{\mathcal{T}}
\newcommand{\TTa}{\TT^{\alpha}}
\newcommand{\TTb}{\TT^{\beta}}
\newcommand{\TTab}{\hat{\TT}^{\alpha,\beta}}
\newcommand{\TTs}{\check{\TT}^{\alpha,\beta}}
\newcommand{\cA}{\mathcal{A}}
\newcommand{\cB}{\mathcal{B}}
\newcommand{\cC}{\mathcal{C}}
\newcommand{\cCa}{\cC^\alpha}
\newcommand{\cCb}{\cC^\beta}
\newcommand{\cCs}{\check{\cC}^{\alpha,\beta}}
\newcommand{\Sab}{\mathcal{S}^{\alpha, \beta}}
\newcommand{\sigmaa}{\sigma^\alpha}
\newcommand{\sigmab}{\sigma^\beta}
\newcommand{\gammaa}{\check{\gamma}^\alpha}
\newcommand{\gammab}{\check{\gamma}^\beta}
\newcommand{\Pa}{P^{\alpha}}
\newcommand{\Pb}{P^{\beta}}
\newcommand{\Pab}{\hat{P}^{\alpha,\beta}}
\newcommand{\Ps}{\check{P}^{\alpha,\beta}}
\newcommand{\cP}{\mathcal{P}}
\newcommand{\tcP}{\tilde{\mathcal{P}}}
\newcommand{\cD}{\mathcal{D}}
\newcommand{\cR}{\mathcal{R}}

\begin{document}

\title[Concave corner Toeplitz operators and corner states]{Toeplitz operators on concave corners and topologically protected corner states}
\author[S. Hayashi]{Shin Hayashi}
\address{Mathematics for Advanced Materials-OIL c/o AIMR Tohoku University, National Institute of Advanced Industrial Science and Technology, 2-1-1 Katahira, Aoba, Sendai 980-8577, Japan}
\email{{\tt shin-hayashi@aist.go.jp}}

\maketitle

\begin{abstract}
We consider Toeplitz operators defined on a concave corner-shaped subset of the square lattice.
We obtain a necessary and sufficient condition for these operators to be Fredholm.
We further construct a Fredholm concave corner Toeplitz operator of index one.
By using this, a relation between Fredholm indices of quarter-plane and concave corner Toeplitz operators is clarified.
As an application, topological invariants and corner states for some bulk-edges gapped Hamiltonians on two-dimensional (2-D) class AIII and 3-D class A systems with concave corners are studied.
Explicit examples clarify that these topological invariants depend on the shape of the system.
We discuss the Benalcazar--Bernevig--Hughes' 2-D Hamiltonian and see that there still exists topologically protected corner states even if we break some symmetries as long as the chiral symmetry is preserved.
\end{abstract}

\setcounter{tocdepth}{2}
\tableofcontents

\section{Introduction}

Toeplitz operators and its index theory, which have been intensively studied in mathematics, are known to play an important role also in condensed matter physics.
In this paper, we consider Toeplitz operators defined on a concave corner-shaped subset of the square lattice and study its index theory.
We then apply these results to the study of topologically protected corner states on systems with codimension-two convex and concave corners.
Benalcazar--Bernevig--Hughes' 2-D model, which leads to the recent active study of higher-order topological insulators, is also studied from this viewpoint.

The topology of gapped Hamiltonians is known to be interesting from a physical point of view \cite{PS16}.
One important aspect of {\em topological insulators} is the existence of topologically protected edge states while its bulk is gapped.
For a quantum Hall system, its topological invariant, known as the TKNN number \cite{TKNN82}, is defined as the first Chern number of the complex vector bundle (called the Bloch bundle) over the two-dimensional torus (called the Brillouin torus).
Such edge states appear corresponding to this topology.
This relation is proved by Hatsugai \cite{Hat93b} and is called the {\em bulk-edge correspondence}.
Kellendonk--Richter--Schulz-Baldes explained this correspondence as an index theory for Toeplitz operators \cite{KRSB02,KRSB00} and generalized it to disordered systems by using the noncommutative geometric technique developed by Connes and Bellissard \cite{BvES94,Connes94}.
Specifically, $K$-theory and index theory applied to the Toeplitz extension of the rotation $C^*$-algebra explains the bulk-edge correspondence for quantum Hall systems.

Apart from these studies, Toeplitz algebras associated with subsemigroups of abelian groups have been much studied \cite{BS06,CD71,Cu17,Do73}.
A cone of the square lattice is an example of such subsemigroups.
Toeplitz operators defined on cones that appear as an intersection of two half-planes are called {\em quarter-plane Toeplitz operators} \cite{DH71,Ji95,Pa90,Sim67}.
Douglas--Howe studied these operators on a quarter-plane of a special shape by using the tensor product structure of the quarter-plane Toeplitz algebra \cite{DH71}.
In this special case, Coburn--Douglas--Singer obtained an index formula to express a Fredholm index of a quarter-plane Toeplitz operator in a topological manner  \cite{CDS72}.
Park further developed Douglas--Howe's technique to the case of general quarter-planes \cite{Pa90}.
Combined with Jiang's construction of Fredholm quarter-plane Toeplitz operators \cite{Ji95}, boundary homomorphisms of $K$-theory for $C^*$-algebras associated with Park's short exact sequence are computed.
In this paper, we regard these cones (quarter-planes) as models of {\em convex corners}.

Since real materials have various shapes, to study the topology of Hamiltonians on systems of various shapes is a natural direction for further research.
In \cite{Hayashi2}, the index theory for quarter-plane Toeplitz operators is applied to the topological study of some gapped Hamiltonians on systems with codimension-two convex corners.
It is shown that for gapped Hamiltonians that are gapped not just on the bulk but also on two edges, there exists a topological invariant that is related to corner states.
In this paper, we refer this relation to the {\em bulk-edge and corner correspondence} \cite{Hayashi2}.
These results are obtained by applying $K$-theory for $C^*$-algebras for the following quarter-plane Toeplitz extension obtained by Douglas--Howe and Park in \cite{DH71,Pa90} (all symbols are defined in the main body of this paper):
\begin{equation}\label{seq1}
0 \to K(\HHab) \to \TTab \overset{\hat{\gamma}}{\to} \Sab \to 0.
\end{equation}
The topological invariant for such a gapped bulk-edges Hamiltonian is defined as an element of some $K$-group of a $C^*$-algebra, and a boundary homomorphism of the six-term exact sequence associated with some short exact sequence of $C^*$-algebras relates these two.
Moreover, in \cite{Hayashi2}, a nontrivial example is obtained by using some tensor product construction.

Recently, topologically protected corner states are intensively studied in condensed matter physics \cite{BBH17a,HWK17,KPVW18} under the name of {\em higher-order topological insulators} \cite{Frank}.
A trigger seems to be the Benalcazar--Bernevig--Hughes' paper \cite{BBH17a}.
They considered a specific 2-D (resp. 3-D) Hamiltonian on a square (resp. cube)-shaped domain.
This system has four codimension-two (resp. eight codimension-three) convex corners of the special shape.
It turns out that this system has corner states.
In order to characterized these higher order phases, they proposed topological quantities named {\em nested Wilson loops}.
On these studies, a role of some spatial symmetries is rather stressed \cite{BBH17a,HF18,Frank}

In this paper, we first study Toeplitz operators defined on a {\em concave corner-shaped subset} of the square lattice $\Z^2$.
Such a concave corner appears as a union of two half-planes.
We consider the $C^*$-algebra $\TTs$ generated by the Toeplitz operators obtained by compressing the translation operators on $\Z^2$ onto the concave corner-shaped subset and show an extension of the following form (Theorem~\ref{main}):
\begin{equation*}
0 \to K(\HHs) \to \TTs \overset{\check{\gamma}}{\to} \Sab {\to} 0.
\end{equation*}
As a result, a necessary and sufficient condition for Fredholmness of concave corner Toeplitz operators is obtained (Theorem~\ref{Fredholm}).
Further, we construct a nontrivial example of Fredholm concave corner Toeplitz operators of index one (Theorem~\ref{construction}).
Comparing them with Jiang's result \cite{Ji95}, a relation between index theory for Toeplitz operators on convex corners (quarter-planes) and concave corners is clarified (Corollary~\ref{relation}).
This result leads to a Coburn--Douglas--Singer-type index formula for Fredholm concave corner Toeplitz operators when the concave corner is of a special shape (Corollary~\ref{concaveCDS}).
In the case of quarter-planes, a linear splitting of the sequence (\ref{seq1}) is constructed by compressing half-plane Toeplitz operators onto quarter-planes \cite{Pa90}.
However, when we study concave corners, they are a subset of neither half-planes nor subsemigroups of $\Z^2$, so compressions do not, at least directly, give a linear splitting.
This is one technical difference between convex and concave cases, and so we adopt a slightly different approach although some discussions of previous results \cite{Ji95,Pa90} still technically apply in concave cases.
We first construct explicitly a rank-one projection as an element of the algebra $\TTs$ and show that the compact operator algebra is contained in this algebra (Proposition~\ref{contain}).
We then show that the quotient algebra $\TTs/K(\HHs)$ is isomorphic to the algebra $\Sab$ (Proposition~\ref{prop2}).
The surjectivity of the homomorphism $\check{\gamma}$ is proved by using the surjectivity of the homomorphism $\hat{\gamma}$ proved in \cite{Pa90} and specifying a dense subalgebra of $\Sab$ (Lemma~\ref{dense}).

We next apply these results to the study of topologically protected corner states.
In \cite{Hayashi2}, only $3$-D class A systems with codimension-two convex corners are discussed, where the short exact sequence of Theorem~\ref{main} enables us to examine such corner states of systems with concave corners.
In this paper, we mainly study $2$-D class AIII systems with codimension-two (convex and concave) corners.
We consider Hamiltonians on the square lattice and assume that they are gapped at zero, not just on the bulk but also on two edges.
For such gapped Hamiltonians, we define a topological invariant as an element of some $K$-group (Definition~\ref{gappedinvAIII}).
We also define another topological invariant for a corner Hamiltonian that is related to corner states (Definition~\ref{gaplessinvAIII}) and show a relation between these two invariants (Theorem~\ref{BECCAIII}).
Integer-valued numerical corner invariants are defined by using traces on $K(\HHab)$ and $K(\HHs)$.
When we consider two edges, we can associate convex and concave corners (see Fig.~\ref{convconc}).
Correspondingly, we can define two numerical corner invariants under our assumption.
We show that these two numerical corner invariants are different by the multiplication by $-1$ (Theorem~\ref{minusAIII}).
Through this relation, the Coburn--Douglas--Singer index formula \cite{CDS72} and its concave corner analogue (Corollary~\ref{concaveCDS}) gives a topological method to compute numerical corner invariants from gapped bulk-edges Hamiltonians.
We also see that if the rank of the space of the internal degree of freedom is two, then our corner topological invariants are necessarily zero (Proposition~\ref{remrank}).
Thus, in order to find a nontrivial example, its rank must be greater than or equal to four.

We further give a construction of explicit examples by using tensor products, as in \cite{Hayashi2}.
We construct some gapped Hamiltonians from two Hamiltonians of $1$-D class AIII (conventional) topological insulators, and the numerical convex corner invariant is given as a product of topological numbers of these two (Theorem~\ref{prodthmAIII}).
By using this construction, we provide an explicit example of Hamiltonians with nontrivial convex and concave corner invariants (Sect.~$5$).
This example clarifies that these corner invariants may change depending on the shape of the system.
Actually, the example discussed there corresponds to the 2-D Hamiltonian discussed by Benalcazar--Bernevig--Hughes in \cite{BBH17a} (Equation (6) of \cite{BBH17a}.
We refer this model to the {\em 2-D BBH model}) when we take parameters in some specific way.
Based on the chiral symmetry, we define an integer-valued topological invariant for the 2-D BBH model and compute it.
The bulk-edge and corner correspondence gives another explanation of the existence of topologically protected corner states for this model.
While a role of spatial symmetries is much discussed in studies of higher-order topological insulators \cite{BBH17a}, our method does not require any spatial symmetry.
Through an example, we see that topologically protected corner states remain even if we break some symmetries which the 2-D BBH model originally have as long as the chiral symmetry is preserved.
Some corresponding results in the case of 3-D class A systems are also collected in Sect.~$4.2$.

This paper is organized as follows.
In Sect.~$2$, we define concave corner Toeplitz operators and introduce the $C^*$-algebra $\TTs$ generated by these operators. In this section, we show a short exact sequence and obtain a necessary and sufficient condition for concave corner Toeplitz operators to be Fredholm.
In Sect.~$3$, we construct an explicit example of a concave corner Fredholm Toeplitz operator of index one and collects some of its consequences.
In Sect.~$4$, we apply these result to the study of topologically protected corner states.
We mainly treat 2-D class AIII systems, though the results for 3-D class A systems are also collected.
In Sect.~$5$, we consider an explicit example of 2-D class AIII Hamiltonian whose corner invariant is nontrivial on a system with a codimension-two (convex and concave) corner.
We also discuss the 2-D BBH model from our viewpoint there.

\section{Concave corner Toeplitz algebras and their extension}
In this paper, we mainly consider concave corners, that is, corners whose angles are strictly greater than $\pi$.
In particular, we study an index theory for Toeplitz operators defined on concave corners.
In this section, we define such operators and study their properties.
Specifically, we consider a $C^*$-algebra generated by concave corner Toeplitz operators and show a short exact sequence that clarifies a necessary and sufficient condition for these operators to be Fredholm.
In this paper, we use only basics about $K$-theory for $C^*$-algebras.
Details can be found in \cite{Bl98,HR00,Mur90,RLL00}, for example.

\subsection{Setup}
Let $\HH$ be the Hilbert space $l^2(\Z^2)$.
For a pair of integers $(m, n)$, let ${ e_{m,n}}$ be the element of $\HH$ that is $1$ at $(m,n)$ and $0$ elsewhere.
For $(m, n) \in \Z^2$, let $M_{m,n} \colon \HH \to \HH$ be the translation operator defined by $(M_{m,n}\varphi)(k,l) = \varphi(k-m, l-n).$\footnote{Note that our choice of translation direction is the same as \cite{Ji95} and different from \cite{Pa90}. In our definition, $M_{m,n}{  e_{s,t}} = {  e_{s+m, t+n}}$ holds.}
We choose real numbers $\alpha < \beta$, and let $\HHa$ and $\HHb$ be the closed subspaces of $\HH$ spanned by $\{ {  e_{m,n}} \mid -\alpha m + n \geq 0 \}$ and $\{ {  e_{m,n}} \mid -\beta m + n \leq 0 \}$, respectively.
$\HHa$ and $\HHb$ model half-planes distinguished by lines $y = \alpha x$ and $y = \beta x$ (see the left-hand side of Fig.~\ref{convconc}).
We here consider two models of spaces with codimension-two boundaries, which we call {\em corners}.
One is an intersection of two half-planes, and the other is a union of these two.
We refer to these two as a {\em convex corner} and a {\em concave corner}, respectively\footnote{The square lattice $\Z^2$ is naturally embedded in the Euclidean space $\R^2$. As a subset of $\R^2$, what we called convex corners are {\em not} convex sets. We here use the words {\em convex} and {\em concave} just to distinguish the two models of corners indicated in Fig.~\ref{convconc}.} (see Fig.~\ref{convconc}).
Specifically, let $\hat{\Sigma} := \{ (x,y) \in \Z^2 \mid -\alpha x + y \geq 0 \ \text{and} -\beta x + y \leq 0 \}$,
and let $\HHab$ be the closed subspace of $\HH$ spanned by elements in the set $\{ \e_{x,y} \mid (x,y) \in \hat{\Sigma} \}$.
Note that the Hilbert space $\HHab$ is intersection $\HHa \cap \HHb$ of $\HHa$ and $\HHb$.
We regard $\HHab$ as a model of a convex corner.
Let $\Pab$ be the orthogonal projection of $\HH$ onto $\HHab$.
Note that $\Pab = \Pa \Pb = \Pb \Pa$.
Let $\check{\Sigma} := \{ (x,y) \in \Z^2 \mid -\alpha x + y \geq 0 \ \text{or} -\beta x + y \leq 0 \}$,
and let $\HHs$ be the closed subspace of $\HH$ spanned by elements in the set $\{ \e_{x,y} \mid (x,y) \in \check{\Sigma} \}$.
We regard $\HHs$ as a model of a concave corner.
Let $\Ps$ be the orthogonal projection of $\HH$ onto $\HHs$.
Note that $\Ps = \Pa + \Pb - \Pa\Pb$.
In what follows, we consider operators on these Hilbert spaces.
The real numbers $\alpha$ and $\beta$ correspond to the slope of two edges (Fig.~\ref{convconc}).
We can take $\alpha = - \infty$ or $\beta = +\infty$, but not both (if $\alpha = - \infty$ and $\beta = +\infty$, the ``corner'' will be the ``edge'').
If we fix $\alpha$ and $\beta$, we can consider two types of corners, that is, convex and concave corners.
In this paper, we treat both of these cases\footnote{In order to distinguish these two cases, we use {\em hat} ``$\wedge$'' for objects associated with convex corners and {\em check} ``$\vee$'' for those with concave corners (e.g., $\HHab$ and $\HHs$).}.

\begin{remark}\label{othercases}
In the main body of this paper, we just treat the case in which the corner (or edges) includes lattice points on lines $y = \alpha x$ and $y = \beta x$.
We can consider variants that do not contain these points.
For these cases, the results of this paper still hold.
Some results in these cases are collected in the appendix of this paper.
\end{remark}
\begin{figure}
\centering
\includegraphics[width=110mm]{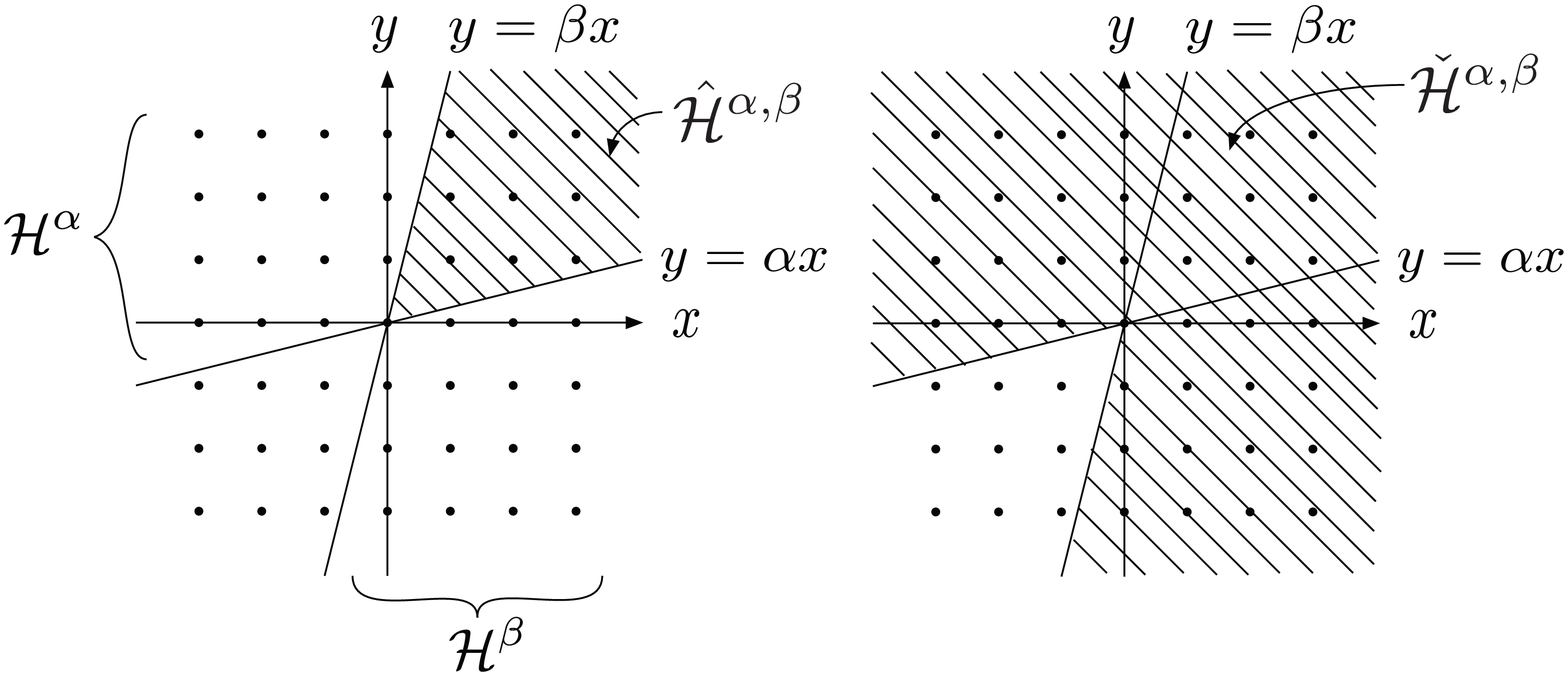}
\caption{A convex corner (left) and a concave corner (right) correspond to shaded area}
\label{convconc}
\end{figure}
The {\em quarter-plane Toeplitz $C^*$-algebra} \cite{DH71,Pa90} is defined to be the $C^*$-subalgebra $\TTab$ of $B(\HHab)$ generated by $\{ \Pab M_{m,n} \Pab \mid (m,n) \in \Z^2 \}$.
Similarly, we define the {\em concave corner Toeplitz $C^*$-algebra} to be the $C^*$-subalgebra $\TTs$ of $B(\HHs)$ generated by $\{ \Ps M_{m,n} \Ps \mid (m,n) \in \Z^2 \}$.
We also define the {\em half-plane Toeplitz $C^*$-algebras} $\TTa$ and $\TTb$ to be $C^*$-subalgebras of $B(\HHa)$ and $B(\HHb)$ generated by $\{ P^\alpha M_{m,n} P^\alpha \mid (m,n) \in \Z^2 \}$ and $\{ P^\beta M_{m,n} P^\beta \mid (m,n) \in \Z^2 \}$, respectively.
Let $\cCa$, $\cCb$ and $\cCs$ be the commutator ideals of $\TTa$, $\TTb$ and $\TTs$, respectively.
As is shown in \cite{CD71}, we have surjective $*$-homomorphisms $\sigmaa \colon\TTa \to C(\T^2)$ and $\sigmab \colon \TTb \to C(\T^2)$ that map $P^\alpha M_{m,n} P^\alpha$ to $\chi_{m,n}$ and $P^\beta M_{m,n} P^\beta$ to $\chi_{m,n}$, respectively, where $\chi_{m,n}(\xi,\eta) = \xi^m \eta^n$.
As in \cite{Pa90}, we define a $C^*$-algebra $\mathcal{S}^{\alpha, \beta}$ to be the pullback of $\TTa$ and $\TTb$ along $C(\T^2)$, that is,
$\Sab := \{ (T^\alpha, T^\beta) \in \TTa \oplus \TTb \mid \sigmaa(T^\alpha) = \sigmab(T^\beta) \}$.
As is shown in \cite{Pa90}, we have surjective $*$-homomorphisms $\hat{\gamma}^\alpha \colon \TTab \to \TTa$ and $\hat{\gamma}^\beta \colon \TTab \to \TTb$ that map $\Pab M_{m,n} \Pab$ to $\Pa M_{m,n} \Pa$ and $\Pab M_{m,n} \Pab$ to $\Pb M_{m,n} \Pb$, respectively.
By using these two, we obtain surjective $*$-homomorphism $\hat{\gamma} \colon \TTab \to \Sab$ given by $\hat{\gamma}(T) = (\hat{\gamma}^\alpha(T), \hat{\gamma}^\beta(T))$.
We write $p^\alpha \colon \Sab \to \TTa$ and $p^\beta \colon \Sab \to \TTb$ for the $*$-homomorphisms given by projections onto each component.
\vspace{-2mm}
\begin{equation}\label{Sab}
\vcenter{
\xymatrix{
\Sab \ar[r]^{p^\beta} \ar[d]_{p^\alpha}& \TTb \ar[d]^{\sigmab}
\\
\TTa \ar[r]^{\sigmaa} & C(\T^2).
}}\vspace{-1mm}
\end{equation}
We write $\sigma$ for the composition $\sigmaa \circ p^\alpha = \sigmab \circ p^\beta$.

Note that the dense subalgebras of $\TTa$, $\TTb$, $\TTab$ and $\TTs$ consist of the following operators:
\vspace{-2mm}
\begin{equation}\label{densea}
	\text{For} \ \ \TTa  \ \colon \ \sum_{i=1}^l c_i \Pa M_{m_{i0}, n_{i0}} \biggl( \prod_{j=1}^{k_i} \Pa M_{m_{ij}, n_{ij}} \biggl) \Pa,
\end{equation}
\vspace{-2mm}
\begin{equation}\label{denseb}
	\text{For} \ \ \TTb \ \colon \ \sum_{i=1}^l c_i \Pb M_{m_{i0}, n_{i0}} \biggl( \prod_{j=1}^{k_i} \Pb M_{m_{ij}, n_{ij}} \biggl) \Pb,
\end{equation}
\vspace{-2mm}
\begin{equation}\label{denseab}
	\text{For} \ \ \TTab \ \colon \ \sum_{i=1}^l c_i \Pab M_{m_{i0}, n_{i0}} \biggl( \prod_{j=1}^{k_i} \Pab M_{m_{ij}, n_{ij}} \biggl) \Pab,
\end{equation}
\vspace{-2mm}
\begin{equation}\label{denses}
	\text{For} \ \ \TTs \ \colon \ \sum_{i=1}^l c_i \Ps M_{m_{i0}, n_{i0}} \biggl( \prod_{j=1}^{k_i} \Ps M_{m_{ij}, n_{ij}} \biggl) \Ps,
\end{equation}
where $c_i \in \C$.

\begin{lemma}\label{dense}
A dense subalgebra of $\Sab$ consists of the pairs of operators of the following form:
\begin{equation}\label{densess}
	\biggl( \sum_{i=1}^l \hspace{-0.5mm} c_i \Pa \hspace{-0.5mm} M_{m_{i0}\hspace{-0.3mm}, n_{i0}} \hspace{-0.5mm} \biggl(\prod_{j=1}^{k_i} \hspace{-0.5mm} \Pa \hspace{-0.5mm} M_{m_{ij} \hspace{-0.3mm},n_{ij}} \hspace{-1mm} \biggl) \hspace{-0.5mm} \Pa \hspace{-0.5mm}, \hspace{-0.5mm} \sum_{i=1}^l \hspace{-0.5mm} c_i \Pb \hspace{-0.5mm} M_{m_{i0}\hspace{-0.3mm}, n_{i0}} \hspace{-0.5mm} \biggl( \prod_{j=1}^{k_i} \hspace{-0.5mm} \Pb \hspace{-0.5mm} M_{m_{ij}\hspace{-0.3mm}, n_{ij}} \biggl) \Pb \hspace{-0.5mm} \biggl)
\end{equation}
where $c_i \in \C$.
\end{lemma}
\begin{proof}
As is shown in \cite{Pa90}, we have a surjective $*$-homomorphism $\hat{\gamma} \colon \TTab \rightarrow \Sab$.
An image of a dense subalgebra of the algebra $\TTab$ under the surjective $*$-homomorphism $\hat{\gamma}$ is a dense subalgebra of $\Sab$.
A dense subalgebra of $\TTab$ consists of operators of the form (\ref{denseab}).
For an operator $\hat{T}$ of the form (\ref{denseab}),
\begin{gather*}
	\hat{\gamma}(\hat{T}) = (\hat{\gamma}^\alpha(\hat{T}), \hat{\gamma}^\beta(\hat{T}))=\\
		\biggl( \sum_{i=1}^l \hspace{-0.5mm} c_i \Pa \hspace{-0.5mm} M_{m_{i0}\hspace{-0.3mm}, n_{i0}} \hspace{-0.5mm} \biggl(\prod_{j=1}^{k_i} \hspace{-0.5mm} \Pa \hspace{-0.5mm} M_{m_{ij} \hspace{-0.3mm},n_{ij}} \hspace{-1mm} \biggl) \hspace{-0.5mm} \Pa \hspace{-0.5mm}, \hspace{-0.5mm} \sum_{i=1}^l \hspace{-0.5mm} c_i \Pb \hspace{-0.5mm} M_{m_{i0}\hspace{-0.3mm}, n_{i0}} \hspace{-0.5mm} \biggl( \prod_{j=1}^{k_i} \hspace{-0.5mm} \Pb \hspace{-0.5mm} M_{m_{ij}\hspace{-0.3mm}, n_{ij}} \biggl) \Pb \hspace{-0.5mm} \biggl).\vspace{-2mm}
\end{gather*}
Thus, the pairs of operators of this form compose a dense subalgebra of $\Sab$.
\end{proof}

\subsection{Surjective $*$-homomorphisms from $\TTs$ to $\TTa$ and $\TTb$}
In this subsection, we construct $*$-homomorphisms from $\TTs$ to $\TTa$ and $\TTb$.
We basically follow the proof of Proposition~$1.2$ of \cite{Pa90}, which treats convex corners, but some points should be modified in our concave case.
We first prepare the following lemma.

\begin{lemma}\label{lem2}
Let $\{ (m_i, n_i ) \}$ be a finite collection of pairs of integers.
Then, there exists a pair of integers $(r,s)$ such that, for all $i$,
\begin{itemize}
	\item $-\alpha (m_i - r)  + (n_i - s) \geq 0$ if and only if $-\alpha m_i + n_i \geq 0$,
	\item $-\beta(m_i - r) + (n_i - s) > 0$.
\end{itemize}
\end{lemma}

\begin{proof}
We choose $\epsilon > 0$, $M > 0$ so that $\epsilon < \min \{ \alpha m_i - n_i \mid -\alpha m_i + n_i < 0 \}$ and $-M \leq \min \{ -\beta m_i + n_i \}$.
Then, it suffices to show that there exist some integers $r$ and $s$ such that
\begin{equation*}
	0 \leq \alpha r -s < \epsilon \qquad \text{and} \qquad -\beta r + s < -M.
\end{equation*}

As in \cite{Pa90}, we here use the following result contained in \cite{HW08}: there exists a positive integer $r$ and an integer $s$ such that
\begin{equation*}
	0 \leq \alpha - \frac{s}{r} < \frac{1}{r^2} \qquad \text{and} \qquad r > \max \biggl\{ \frac{1}{\epsilon}, \frac{M}{\beta - \alpha} \biggl\}.
\end{equation*}
For such $r$ and $s$, we have
\vspace{-2mm}
\begin{equation*}
	0 \leq \alpha r - s < \frac{1}{r} < \epsilon,
	\vspace{-1mm}
\end{equation*}
and
\begin{equation*}\vspace{-1mm}
	-\beta r + s = -(\beta - \alpha)r + (-\alpha r + s) \leq -(\beta - \alpha) r \leq -M.
\end{equation*}
as desired.
\end{proof}

\begin{proposition}\label{prop1}
There exists surjective $*$-homomorphisms
\begin{equation*}
	\gammaa \colon \TTs \to \TTa, \ \ \gammab \colon \TTs \to \TTb.
\end{equation*}
\end{proposition}

\begin{proof}
For $\check{T} = \sum_{i=1}^l c_i \Ps M_{m_{i0}, n_{i0}} \bigl( \prod_{j=1}^{k_i} \Ps M_{m_{ij}, n_{ij}} \bigl) \Ps$, we set
\vspace{-2mm}
\begin{equation}\label{image}
	\gammaa(\check{T}) := \sum_{i=1}^l c_i \Pa M_{m_{i0}, n_{i0}} \biggl( \prod_{j=1}^{k_i} \Pa M_{m_{ij}, n_{ij}} \biggl) \Pa.
\vspace{-2mm}
\end{equation}
and
\vspace{-2mm}
\begin{equation*}
	\gammab(\check{T}) := \sum_{i=1}^l c_i \Pb M_{m_{i0}, n_{i0}} \biggl( \prod_{j=1}^{k_i} \Pb M_{m_{ij}, n_{ij}} \biggl) \Pb.
\end{equation*}
To show that $\gammaa$ and $\gammab$ are well-defined and extend to $*$-homomorphisms on $\TTs$, it is sufficient to show $\| \gammaa(\check{T}) \| \leq \| \check{T} \|$ and $\| \gammab(\check{T}) \| \leq \| \check{T} \|$.
We here discuss $\gammaa$ only. The result for $\gammab$ is proved in almost the same way.

Let $\epsilon > 0$.
We take $f \in \HHa$ such that $f$ has a finite support, $\| f \|= 1$ and $\| \gammaa(\check{T}) \| \leq \| \gammaa(\check{T}) f \| + \epsilon$.
Let $S$ be the union of the set $\supp(f)$ and the following set
\vspace{-2mm}
\begin{equation*}
	\left\{ \biggl( m_0 + \sum_{j=N}^{k_i} m_{ij}, n_0 + \sum_{j=N}^{k_i} n_{ij} \biggl) \ \Biggl\vert
\begin{array}{ll}
 (m_0, n_0) \in \supp(f),\\
  \ 0 \leq i \leq l, 0 \leq N \leq k_i
 \end{array}
\right\}.
\end{equation*}
The set $S$ is a finite subset of $\Z^2$.
Applying Lemma~\ref{lem2} to the set $S$, we obtain a pair $(r, s)$ of integers such that for any $(m, n) \in S$, we have
\begin{itemize}
	\item $-\alpha (m-r)  + (n-s) \geq 0$ if and only if $-\alpha m + n \geq 0$,
	\item $-\beta(m-r) + (n - s) > 0$.
\end{itemize}
This leads to the following relation:
\begin{itemize}
	\item $\gammaa(\check{T})M_{-r, -s} f = M_{-r, -s} \gammaa(\check{T}) f$,
	\item $\gammaa(\check{T})M_{-r, -s} f = \check{T} M_{-r, -s} f$.
\end{itemize}
By using this, we have
\begin{align*}
	\| \gammaa(\check{T}) \| &\leq \| \gammaa(\check{T})f \| + \epsilon 
	  = \| M_{-r,-s} \gammaa(\check{T}) f \| + \epsilon \\
	  &= \| \gammaa(\check{T}) M_{-r,-s} f \| + \epsilon 
	  = \| \check{T} M_{-r,-s} f \| + \epsilon \\
	  &\leq \| \check{T} \| \| M_{-r,-s} \| \| f \| + \epsilon
	  = \| \check{T} \| + \epsilon.
\end{align*}
Thus, $\| \gammaa(\check{T}) \| \leq \| \check{T} \|$ holds.

Since $\gammaa$ is a $*$-homomorphism and operators of the form $(\ref{image})$ compose a dense subalgebra of $\TTa$, the map $\gammaa$ is surjective.
\end{proof}

Since $\gammaa \circ \sigmaa = \gammab \circ \sigmab$, we have a $*$-homomorphism
$\check{\gamma} \colon \TTs \to \Sab$ given by $\check{\gamma}(\check{T}) = (\gammaa(\check{T}), \gammab(\check{T}))$.
By Lemma~\ref{dense}, the map $\check{\gamma}$ is surjective.

\vspace{-2mm}
\subsection{$K(\HHs) \subset \TTs$}
For $(x,y) \in \check{\Sigma}$, let $p_{x,y}$ be the orthogonal projection of $\HHs$ onto $\C \e_{x,y}$. 
In this subsection, we show the following proposition by constructing explicit rank-one projections contained in the algebra $\TTs$.
\begin{proposition}\label{contain}
$K(\HHs) \subset \TTs$.
Moreover, $K(\HHs)$ is contained in $\Ker \check{\gamma}$.
\end{proposition}
To show this proposition, we employ a trick by Jiang \cite{Ji95}.
We consider the action of $SL(2, \Z)$ onto $\Z^2$.
An action of $g \in SL(2,\Z)$ maps a line through the origin whose slope is $s$ to the line through the origin of possibly different slope. We write $g(s)$ for its slope.
It is shown in Sect.~$1$ of \cite{Ji95} that there is a $g \in SL(2,\Z)$ such that $0 < g(\alpha) \leq \frac{1}{2}$ and $1 \leq g(\beta) < + \infty$.
The action of $g$ induces a unitary isomorphism between Hilbert spaces $\HHs$ and $\check{\HH}^{g(\alpha),g(\beta)}$ and thus induces an isomorphism between $C^*$-algebras $\TTs$ and $\check{\TT}^{g(\alpha),g(\beta)}$ without changing their Fredholm index theory.
Thus, we assume the following condition without loss of generality:
\begin{equation}
	0 < \alpha \leq \frac{1}{2} , \quad 1 \leq \beta < + \infty. \tag{$\dagger$}
\end{equation}
\vspace{-2mm}

For $(m,n) \in \Z^2$, let
$\check{\cP}_{m,n} := \Ps M_{m,n} \Ps M_{-m,-n} \Ps$.
The operator $\check{\cP}_{m,n}$ is a projection contained in $\TTs$ (e.g., the projection $1 - \check{\cP}_{-1,0}$ is explained in Fig. \ref{P-10}).
For $k \in \{ 1, 2, \cdots \}$, let 
\begin{equation*}
	\tcP_k := (1-\check{\cP}_{-1,0})\check{\cP}_{-k-1,0} - (1- \check{\cP}_{-1,0})\check{\cP}_{-k,0} \in \TTs,
\end{equation*}
and let $B_k := \{ (x, y) \in \Z^2 \mid 0 \leq -\alpha x + y < \alpha \ \text{and} \ k\beta < -\beta x + y \leq (k+1)\beta \}$.
Then, $\tcP_k$ is the orthogonal projection of $\HHs$ onto the closed subspace spanned by elements in the set $\{ \e_{x,y} \mid (x, y) \in B_k \}$
(the projection $\tcP_3$ is explained in Fig. \ref{N3}).
\begin{figure}
\centering
    \includegraphics[width=80mm]{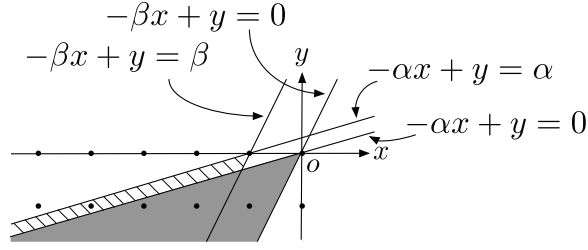}
	\caption{The case of $\frac{1}{4} < \alpha < \frac{1}{3}$ and $1 < \beta < \infty$. $1- \check{\cP}_{-1,0}$ is the orthogonal projection onto closed subspace corresponding to lattice points contained in the shaded area}
	\label{P-10}
\end{figure}
\begin{figure}
\centering
    \includegraphics[width=90mm]{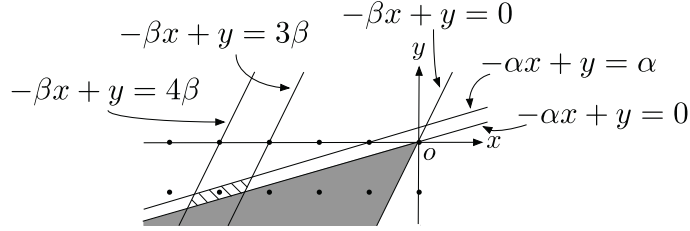}
	\caption{The case of $\frac{1}{4} < \alpha < \frac{1}{3}$ and $1 < \beta < \infty$. $\tcP_3$ is the orthogonal projection $p_{-4,-1}$ onto closed subspace $\C \e_{-4,-1}$.
The set $B_3$ contains just one element, $(-4,-1)$, which corresponds to the lattice point contained in the shaded area}
	\label{N3}	
\end{figure}
For $\alpha$ satisfying the condition ($\dagger$), there exists a unique $N \in \{ 2,3, \cdots \}$ such that $\frac{1}{N+1} < \alpha \leq \frac{1}{N}$.
We show some $\tcP_k$ is a rank-one projection.
The statement is divided into five cases corresponding to the values of $\alpha$ and $\beta$.
\begin{lemma}\label{Ngeneral}
Let $N \geq 2$.
\begin{enumerate}
\renewcommand{\labelenumi}{\arabic{enumi})}
	\item When $\alpha = \frac{1}{2}$ and $\beta = 1$, we have $\tcP_1 = p_{-4,-2}$.
	\item When $N \geq 3$, $\alpha = \frac{1}{N}$ and $\beta = 1$, we have $\tcP_{N-2} = p_{-N,-1}$.
	\item When $\alpha = \frac{1}{N}$ and $1 < \beta < \infty$, we have $\tcP_{N-1} = p_{-N,-1}$.
	\item When $\frac{1}{N+1} < \alpha < \frac{1}{N}$ and $\beta = 1$, we have $\tcP_{N-1} = p_{-N-1,-1}$.
	\item When $\frac{1}{N+1} < \alpha < \frac{1}{N}$ and $1 < \beta < \infty$, we have $\tcP_N = p_{-N-1,-1}$.
\end{enumerate}
\end{lemma}
\begin{proof}
It is sufficient to show that, in each case, the set $B_k$ contains just one element $(x, y)$, where $k$ and $(x, y)$ correspond to subscripts of $\tcP_k$ and $p_{x,y}$ indicated above.
The proof of 1) $\sim$ 5) goes almost in the same way, but note that it is convenient to distinguish the cases of $N = 2$ and $N \geq 3$ also in the case of 3) $\sim$ 5).
We here present just the proof of 5) for the case of $N \geq 3$.

Let $N \geq 3$. We calculate the set $B_N$, i.e., all values of $(x, y) \in \Z^2$ satisfying inequalities $0 \leq -\alpha x + y < \alpha$ and $N\beta < -\beta x + y \leq (N+1)\beta$, and show that it is just one point, $(-N-1, -1)$.
From these two inequalities, we obtain the following inequality:
\vspace{-1mm}
\begin{equation*}
	(-N-1)\frac{\alpha \beta}{\beta - \alpha} \leq y <  (-N+1)\frac{\alpha \beta}{\beta - \alpha} < 0.
\vspace{-1mm}
\end{equation*}
Since $N \geq 3$, $\alpha < \frac{1}{N}$ and $1 < \beta$, the left-hand side is strictly greater than $-2$.
Thus, an integer $y$ should be $-1$.
When $y = -1$, we have $-1 - \frac{1}{\alpha} < x \leq - \frac{1}{\alpha}$.
Since $\frac{1}{N+1} < \alpha < \frac{1}{N}$, $x$ should be $-N-1$.
The point $(-N-1, -1)$ satisfies the desired inequality, and so $B_N = \{ (-N-1, -1) \}$.
\end{proof}
{\em Proof of Proposition~\ref{contain}} \
By using Lemma~\ref{Ngeneral}, we see that the algebra $\TTs$ contains at least a rank-one projection $p_{x,y}$ for some $(x,y) \in \check{\Sigma}$.
For any $(u,v) \in \check{\Sigma}$, we have
\begin{equation*}
	p_{u,v} = ( \Ps M_{u-x,v-y} \Ps) p_{x,y} (\Ps M_{u-x,v-y} \Ps )^* \in \TTs.
\end{equation*}
Thus, $\TTs$ contains rank-one projections $p_{u,v}$ for any $(u,v) \in \check{\Sigma}$ and thus contains operators of the form $\Ps M_{u-x,v-y} \Ps p_{x,y}$ for any $(u,v), (x,y) \in \check{\Sigma}$.
By using this, we can see that every rank-one projection on $\HHs$ is contained in $\TTs$, and thus $\TTs$ contains all finite-rank operators on $\HHs$.
Thus, the inclusion $K(\HHs) \subset \TTs$ holds.

To further show that $K(\HHs)$ is contained in $\Ker \check{\gamma}$, it is sufficient to show that $\check{\gamma}(\tcP_{k}) = (\gammaa(\tcP_{k}), \gammab(\tcP_{k})) = 0$ for $k \geq 1$. 
We have
\begin{gather*}
	\gammaa(\tcP_{k}) = \gammaa(1-\check{\cP}_{-1,0}) \gammaa(\check{\cP}_{-k-1,0} - \check{\cP}_{-k,0}) =\\
		(1 - \Pa M_{-1,0} \Pa M_{1,0} \Pa)(\Pa M_{-k-1,0}\Pa M_{k+1,0}\Pa - \Pa M_{-k,0}\Pa M_{k,0}\Pa).
\end{gather*}
$1 - \Pa M_{-1,0} \Pa M_{1,0} \Pa$ and $\Pa M_{-k-1,0}\Pa M_{k+1,0}\Pa - \Pa M_{-k,0}\Pa M_{k,0}\Pa$ are projections onto closed subspaces spanned by sets $\{ \e_{x,y} \mid 0 \leq -\alpha x +y < \alpha \}$ and $\{ \e_{x,y} \mid k \leq -\alpha x +y < (k+1)\alpha \}$, respectively.
Thus, for $k \geq 1$, we have $\gammaa(\tcP_{k}) = 0$.
We also have $\gammab(\tcP_{k}) = 0$ since $\gammab(1-\check{\cP}_{-1,0}) = 0$.\qed

\subsection{Concave corner Toeplitz extension}
The following is the main theorem of this paper.
\begin{theorem}\label{main}
	There is the following short exact sequence of $C^*$-algebras:
\vspace{-1mm}
\begin{equation}\label{exact}
0 \to K(\HHs) \to \TTs \overset{\check{\gamma}}{\to} \Sab \to 0,
\vspace{-1mm}
\end{equation}
where $K(\HHs)$ is the $C^*$-algebra of compact operators on $\HHs$.
\end{theorem}
In this subsection, we give a proof of this theorem.

\begin{proposition}\label{prop2}
There is a $*$-isomorphism
$\theta \colon \Sab \to \TTs/K(\HHs)$,
that is, on the dense subalgebra of $\Sab$ obtained from Lemma~\ref{dense}, of the following form:
\begin{gather*}
	\theta \biggl( \sum_{i=1}^l \hspace{-0.5mm} c_i \Pa \hspace{-0.5mm} M_{m_{i0}\hspace{-0.3mm}, n_{i0}} \hspace{-0.5mm} \biggl(\prod_{j=1}^{k_i} \hspace{-0.5mm} \Pa \hspace{-0.5mm} M_{m_{ij} \hspace{-0.3mm},n_{ij}} \hspace{-1mm} \biggl) \hspace{-0.5mm} \Pa \hspace{-0.5mm}, \hspace{-0.5mm} \sum_{i=1}^l \hspace{-0.5mm} c_i \Pb \hspace{-0.5mm} M_{m_{i0}\hspace{-0.3mm}, n_{i0}} \hspace{-0.5mm} \biggl( \prod_{j=1}^{k_i} \hspace{-0.5mm} \Pb \hspace{-0.5mm} M_{m_{ij}\hspace{-0.3mm}, n_{ij}} \biggl) \Pb \hspace{-0.5mm} \biggl)
		\\
		=
			\biggl[ \sum_{i=1}^l c_i \Ps M_{m_{i0}, n_{i0}} \biggl( \prod_{j=1}^{k_i} \Ps M_{m_{ij}, n_{ij}} \biggl) \Ps \biggl].
\end{gather*}
Its inverse is the $*$-homomorphism induced by $\check{\gamma}$.
\end{proposition}
\begin{proof}
Let $\check{T}$ be an element of $\TTs$ of the form (\ref{denses}) and $(T^\alpha, T^\beta)$ be an element of $\Sab$ of the form (\ref{densess}).
By Lemma~\ref{dense}, such operators form a dense subalgebra of $\Sab$.
We define $\theta(T^\alpha, T^\beta) := [\check{T}] \in \TTs/K(\HHs)$.
To show the well-definedness of $\theta$ and that $\theta$ extends to a $*$-homomorphism on $\Sab$, it suffices to show the following inequality:
\vspace{-1mm}
\begin{equation*}
	\| \theta(T^\alpha, T^\beta) \| \leq \| (T^\alpha, T^\beta) \|.
\end{equation*}

We relabel the set $\check{\Sigma}$ as in Fig. \ref{relabel}. This gives an order on the set $\{ e_{x,y} \mid (x,y) \in \check{\Sigma} \}$. Let $P_n$ be the orthogonal projection onto the span of the first $n$ elements. Then, for $[\check{T}] \in \TTs/K(\HHs)$, we have
\begin{equation*}
	\| [\check{T}] \| = \inf_{C \in K(\HHs)} \| \check{T} + C \| = \lim_n \| \check{T}(1 - P_n) \|.
\end{equation*}
The last equality follows since $\{ P_n \}_{n=0}^{\infty}$ is an approximate unit for $K(\HHs)$ (see Theorem $1.7.4$ of \cite{HR00}, for example).
\begin{figure}
\centering
    \includegraphics[width=60mm]{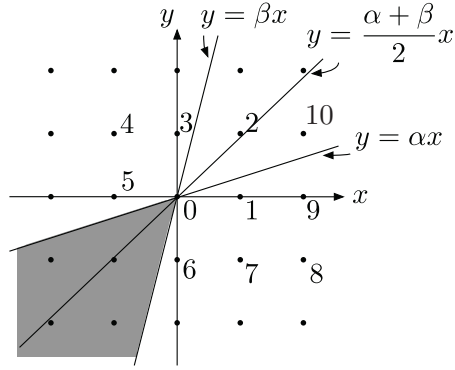}
	\caption{Relabel lattice points as $0,1,2, \cdots$. Divide the set $\check{\Sigma}$ into two parts}
	\label{relabel}
\end{figure}
Further, we divide the set $\check{\Sigma}$ into two parts $\cA$ and $\cB$ by the line $y = \frac{\alpha + \beta}{2}x$, as in Fig. \ref{relabel}.
Specifically, let $\cA := \{ (x,y) \in \check{\Sigma} \mid -\frac{\alpha + \beta}{2}x + y \geq 0 \}$ and $\cB := \check{\Sigma} \setminus \cA$.
Let $M = \| (T^\alpha, T^\beta) \| = \max \{ \| T^\alpha \|, \| T^\beta \| \}$, and let $f \in l^2(\check{\Sigma})$, which has a finite support and satisfies $\| f \| = 1$.
There exists $n_0 \in \N$ such that for any $n \geq n_0$, we have
$\check{T}(1-P_n) f|_\cA = T^\alpha (1-P_n) f|_\cA$ and $\check{T}(1-P_n) f|_\cB = T^\beta (1-P_n) f|_\cB$.
Since the operator $\check{T}$ is of the form (\ref{denses}), we can take such $n_0$ uniformly with respect to $f$.
Thus, for $n \geq n_0$, we have
\begin{align*}
	\| \check{T}(1-P_n) f \| &\leq \| \check{T}(1-P_n) f|_\cA \| +  \| \check{T}(1-P_n) f|_\cB \| \\
		&= \|T^\alpha (1-P_n) f|_\cA \| + \| T^\beta (1 - P_n) f|_\cB \|   \\
		&\leq \|T^\alpha \| \| f|_\cA \| + \| T^\beta \| \| f|_\cB \|   \\
		&\leq M (\| f|_\cA \| + \| f|_\cB \|)
		= M \| f \|
		= M.
\end{align*}
Thus, we have $\| \check{T}(1-P_n) \| \leq M$. By taking $n \to \infty$,  we have $\| \theta(T^\alpha, T^\beta) \| = \| [\check{T}] \| \leq M = \| (T^\alpha, T^\beta) \|$, as desired.

By Proposition~\ref{contain}, $\check{\gamma}$ induces a $*$-homomorphism $\TTs/K(\HHs) \to \Sab$.
By computing on dense subalgebras of $\TTs/K(\HHs)$ and $\Sab$, we can check that this map is an inverse of $\theta$.
Thus, $\theta$ is an isomorphism.
\end{proof}
{\em Proof of Theorem~\ref{main}} \
By Proposition~\ref{contain}, we have a short exact sequence
$0 \to K(\HHs) \to \TTs \to \TTs/K(\HHs) \to 0$.
By Proposition~\ref{prop2}, we have the isomorphism $\theta \colon \Sab \to \TTs/K(\HHs)$.
Combined with these results, we obtain the desired result.\qed
\vspace{2mm}

\noindent
By Theorem~\ref{main}, a necessary and sufficient condition for concave corner Toeplitz operators to be Fredholm is obtained.
\begin{theorem}\label{Fredholm}
An operator $\check{T} \in \TTs$ is Fredholm if and only if $\check{\gamma}(\check{T})$ is invertible in $\Sab$ or, equivalently, if and only if $\gammaa(\check{T})$ and $\gammab(\check{T})$ are invertible in $\TTa$ and $\TTb$, respectively.
\end{theorem}

\section{A Fredholm operator of index one and an index formula}
In this section, we study further concave corner Toeplitz operators from the viewpoint of index theory.
We explicitly construct a Fredholm Toeplitz operator associated with a concave corner whose index is one.
By using this result, we compute some $K$-groups associated with concave corners and boundary homomorphisms associated with the extension (\ref{exact}) of Theorem~\ref{main}.
Moreover, a relation with index theory for quarter-plane Toeplitz operators \cite{Ji95,Pa90} is obtained.
By using this relation, we show some corresponding results obtained previously for quarter-plane Toeplitz operators \cite{CDS72,DH71}.
Especially, a Coburn--Douglas--Singer-type index formula for Fredholm concave corner Toeplitz operators is obtained.

\subsection{A Fredholm operator of index one}
We first construct a Fredholm concave corner Toeplitz operator of index one and compute $K$-groups of some $C^*$-algebras associated with concave corners.

As in \cite{Ji95}, by using the action of $SL(2, \Z)$ onto $\Z^2$, we assume the condition ($\dagger$) without loss of generality.
In this section, we consider the following operator:
\begin{equation*}
	\check{A} := \check{\cP}_{0,1} + M_{1,1}(1 - \check{\cP}_{-1,0}) + M_{1,0}(\check{\cP}_{-1,0} - \check{\cP}_{0,1}).
\end{equation*}
Since $\check{A} = \check{\cP}_{0,1} + \Ps M_{1,1} \Ps (1 - \check{\cP}_{0,-1}) + \Ps M_{1,0} \Ps (\check{\cP}_{-1,0} - \check{\cP}_{0,1})$, the operator $\check{A}$ is an element of the algebra $\TTs$.
The following theorem is the main theorem of this section.
\begin{theorem}\label{construction}
$\check{A}$ is a surjective Fredholm operator whose Fredholm index is $1$.
Its kernel is given as follows:
\begin{enumerate}
\renewcommand{\labelenumi}{\arabic{enumi})}
	\item When $\alpha = \frac{1}{2}$ and $\beta = 1$, $\Ker \check{A} = \C (\e_{-3,-1} -\e_{-2, -1})$.
	\item When $0 < \alpha < \frac{1}{2}$ and $\beta = 1$, $\Ker \check{A} = \C (\e_{-2, 0} - \e_{-1,0})$.
	\item When $\alpha = \frac{1}{2}$ and $1 < \beta < \infty$, $\Ker \check{A} = \C(\e_{-1,0} - \e_{0,0})$.
	\item When $0 < \alpha < \frac{1}{2}$ and $1 < \beta < \infty$, $\Ker \check{A} = \C(\e_{-1,0} - \e_{0,0})$.
\end{enumerate}
Moreover, we have $\check{A} - 1 \in \cCs$.
\end{theorem}
\begin{proof}
To examine the operator $\check{A}$, it is convenient to divide its domain and range as follows.
We divide the set $\check{\Sigma}$ into three parts $\check{\Sigma} = \cD_1 \sqcup \cD_2 \sqcup \cD_3$, where
\begin{equation*}
	\cD_1 = \{ (x,y) \in \check{\Sigma} \mid 0 \leq -\alpha x + y < \alpha \ \text{and} \ 1 < - \beta x + y \},
\end{equation*}
\begin{equation*}
	\cD_2 = \{ (x,y) \in \check{\Sigma} \mid \alpha \leq -\alpha x + y < 1 \ \text{and} \ \beta <  -\beta x + y \},
\end{equation*}
\begin{equation*}
	\cD_3 = \check{\Sigma} \setminus (\cD_1 \sqcup \cD_2).
\end{equation*}
Note that it can be checked that the set $\{ (x,y) \in \check{\Sigma} \mid 0 \leq -\alpha x + y < \alpha \ \text{and} \ 1 < -\beta x + y \leq \beta\}$ is empty under the assumption ($\dagger$).

We also divide $\check{\Sigma}$ in the following way,
$\check{\Sigma} = \cR_1 \sqcup \cR_2 \sqcup \cR_3$, where
\begin{equation*}
	\cR_1 = \{ (x,y) \in \check{\Sigma} \mid 0 \leq -\alpha x + y < 1 - \alpha, \ \text{and} \ 1 < - \beta x + y \},
\end{equation*}
\begin{equation*}
	\cR_2 = \{ (x,y) \in \check{\Sigma} \mid 1 - \alpha \leq -\alpha x + y < 1 \ \text{and} \ 1 <  -\beta x + y \},
\end{equation*}
\begin{equation*}
	\cR_3 = \check{\Sigma} \setminus (\cR_1 \sqcup \cR_2).
\end{equation*}
Note that $\cD_3 = \cR_3$ (see Fig.~\ref{domain} and Fig.~\ref{range}).
\begin{figure}
\centering
    \includegraphics[width=110mm]{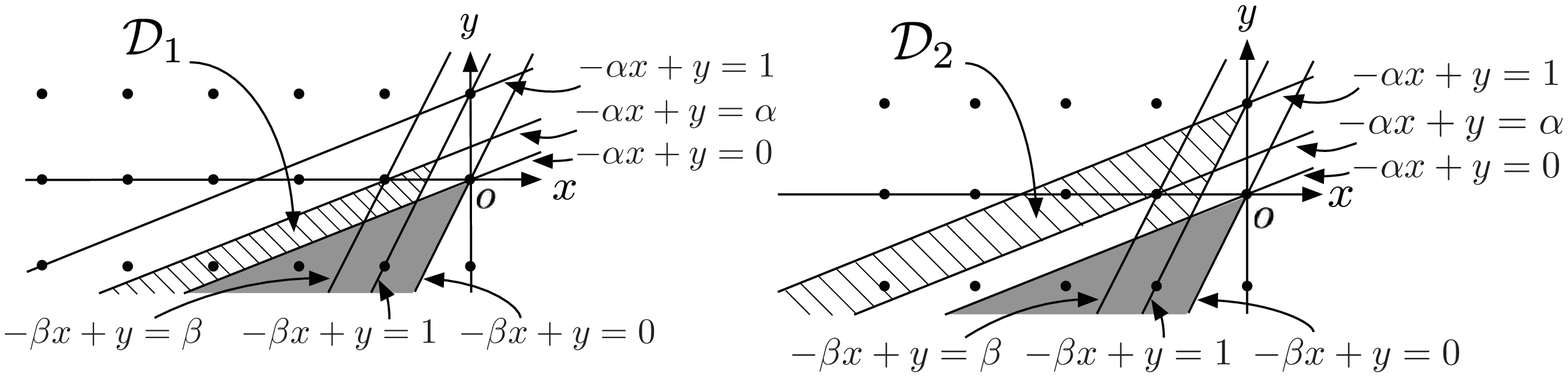}
	\caption{$\cD_1$ and $\cD_2$}
	\label{domain}
\end{figure}
\begin{figure}
\centering
    \includegraphics[width=110mm]{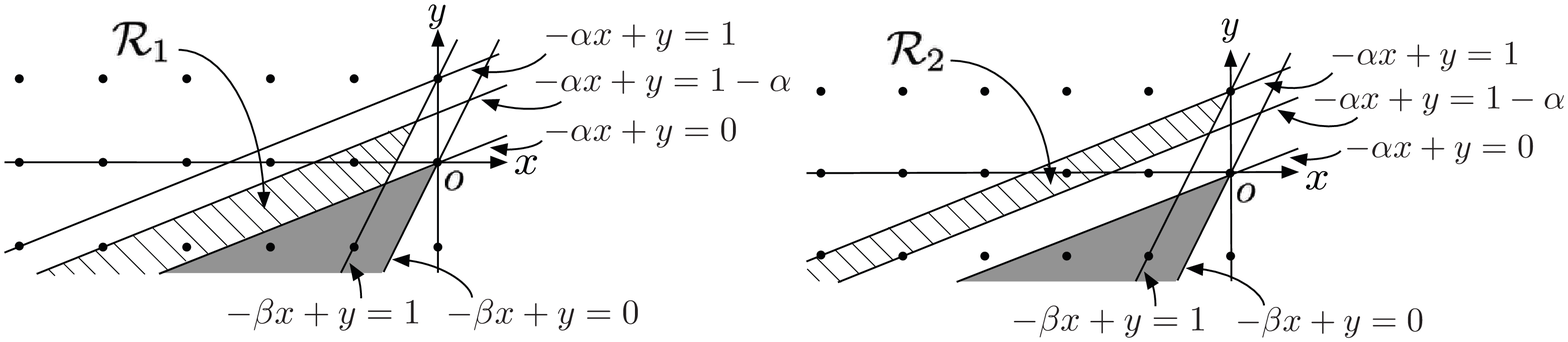}
	\caption{$\cR_1$ and $\cR_2$}
	\label{range}
\end{figure}
First, we have
\begin{equation*}
\check{A} \e_{x,y} =
\left\{
\begin{aligned}
\e_{x+1, y+1} & \hspace{3mm} \text{if} \ (x, y) \in \cD_1, \\
\e_{x+1,y} \ & \hspace{3mm} \text{if} \ (x, y) \in \cD_2, \\
\e_{x,y} \ \ & \hspace{3mm} \text{if} \ (x, y) \in \cD_3. \\
\end{aligned}
\right.
\end{equation*}
Actually, if $(x, y) \in \cD_2$, we have
\begin{align*}
	\check{A} \e_{x,y} &= \check{\cP}_{0,1} \e_{x,y} + M_{1,1}(1 - \check{\cP}_{0,-1}) \e_{x,y}  + M_{1,0}(\check{\cP}_{-1,0} - \check{\cP}_{0,1}) \e_{x,y}\\
		&= 0 + 0 + M_{1,0} \e_{x,y} = \e_{x+1, y},
\end{align*}
and the other cases follow via a similar computation.
$\check{A}$ is surjective since
\begin{equation*}
\e_{x,y} =
\left\{
\begin{aligned}
\check{A} \e_{x-1,y} \ & \hspace{3mm} \text{if} \ (x, y) \in \cR_1,\\
\check{A} \e_{x-1, y-1} & \hspace{3mm} \text{if} \ (x, y) \in \cR_2,\\
\check{A} \e_{x,y} \ \ & \hspace{3mm}\text{if} \ (x, y) \in \cR_3. \\
\end{aligned}
\right.
\end{equation*}
Actually, if $(x, y) \in \cR_2$, that is, $1 - \alpha \leq -\alpha x + y < 1$ and $1 < -\beta x + y$, then $0 \leq -\alpha(x-1) + (y-1) < \alpha$ and $\beta < -\beta (x-1) + (y-1)$.
Thus, $(x-1, y-1) \in \cD_1$, and we have $\check{A} \e_{x-1, y-1} = \e_{x,y}$;
the other cases follow via similar computations.

We next show the following results:
\begin{itemize}
	\item There exist two points $(x_0, y_0)$ and $(x_1, y_1)$ in the set $\check{\Sigma}$ such that $\check{A} \e_{x_0,y_0} = \check{A} \e_{x_1, y_1} = \e_{x_1, y_1}$.
	\item For any point $(u,v) \in \check{\Sigma} \setminus \{ (x_1,y_1) \}$, there exists unique point $(x, y) \in \check{\Sigma} \setminus \{ (x_0, y_0), (x_1,y_1) \}$ such that $\check{A} \e_{x,y} = \e_{u,v}$.
\item These $(x_0, y_0)$ and $(x_1, y_1)$ take the following values:
\begin{enumerate}
\renewcommand{\labelenumi}{\arabic{enumi})}
	\item When $\alpha \hspace{-0.5mm} = \hspace{-0.5mm} \frac{1}{2}$ and $\beta \hspace{-0.5mm} =\hspace{-0.5mm} 1$, $(x_0, y_0) \hspace{-0.5mm}=\hspace{-0.5mm} (-3, -1)$ and $(x_1, y_1) \hspace{-0.5mm}=\hspace{-0.5mm} (-2,-1)$,
	\item When $0 \hspace{-0.5mm}<\hspace{-0.5mm} \alpha \hspace{-0.5mm}<\hspace{-0.5mm} \frac{1}{2}$ and $\beta \hspace{-0.5mm}=\hspace{-0.5mm} 1$, $(x_0, y_0) \hspace{-0.5mm}=\hspace{-0.5mm} (-2, 0)$ and $(x_1, y_1) \hspace{-0.5mm}=\hspace{-0.5mm} (-1,0)$,
	\item When $\alpha \hspace{-0.5mm}=\hspace{-0.5mm} \frac{1}{2}$ and $1 \hspace{-0.5mm}<\hspace{-0.5mm} \beta \hspace{-0.5mm}<\hspace{-0.5mm} \infty$, $(x_0, y_0) \hspace{-0.5mm}=\hspace{-0.5mm} (-1, 0)$ and $(x_1, y_1) \hspace{-0.5mm}=\hspace{-0.5mm} (0, 0)$,
	\item When $0 < \alpha < \frac{1}{2}$ and $1 < \beta < \infty$, $(x_0, y_0) = (-1, 0)$ and $(x_1, y_1) = (0, 0)$.
\end{enumerate}
\end{itemize}
We here prove them only in the case 4), that is, when $0 < \alpha < \frac{1}{2}$ and $1 < \beta < \infty$.
The other cases can be shown in the same way.

When $(x,y) \in \cD_1$, that is, $0 \leq -\alpha x + y < \alpha$ and $\beta < -\beta x + y$, we have
\begin{equation*}
	1 - \alpha \leq -\alpha (x+1) + (y+1) < 1 \ \text{and} \ \beta < -\beta (x+1) + (y+1),
\end{equation*}
and thus, the point $(x+1, y+1)$ is contained in $\cR_2$.
On the other hand, if $(x, y) \in \cR_2$, then $(x-1, y-1) \in \cD_2$.
Thus, there is a bijection
\vspace{-2mm}
\begin{equation*}
	\{ \e_{x,y} \mid (x, y) \in \cD_1 \} \overset{\check{A} = M_{1,1}}{\longrightarrow} \{ \e_{x,y} \mid (x, y) \in \cR_2 \}.
\end{equation*}
We next compute points $(x, y)$ in $\cR_1$ for which $(x+1, y)$ is not contained in $\cR_1$.
Such $(x,y) \in \Z^2$ satisfy $0 \leq -\alpha (x+1) + y < 1 - \alpha$ and $1 - \beta < -\beta(x+1) + y \leq 1$.
There is just one point that satisfies these inequalities, and under the assumption of 4), this point is $(-1, 0)$.
As in the case of $\cD_1$, there is a bijection
\vspace{-2mm}
\begin{equation*}
	\{ \e_{x,y} \mid (x, y) \in \cD_2 \setminus \{ (-1, 0) \} \} \overset{\check{A} = M_{1,0}}{\longrightarrow} \{ \e_{x,y} \mid (x, y) \in \cR_2 \}.
\end{equation*}
The result follows since
\begin{equation*}
	\{ \e_{x,y} \mid (x, y) \in \cD_3 \} \overset{\check{A} = \id}{\longrightarrow} \{ \e_{x,y} \mid (x, y) \in \cR_3 \},
\end{equation*}
is a bijection and $\check{A} \e_{-1,0} = \check{A} \e_{0,0} = \e_{0,0}$.

By applying the method in \cite{CD71} for our subset $\check{\Sigma}$, we obtain the following short exact sequence,
\vspace{-2mm}
\begin{equation*}
	0 \to \cCs \to \TTs \overset{\sigma \circ \check{\gamma}}{\to} C(\T^2) \to 0.
\end{equation*}
Note that the set $\check{\Sigma}$ contains the subsemigroup $\hat{\Sigma}$ of the discrete abelian group $\Z^2$.
Note also that $\hat{\Sigma}$ acts on the set $\check{\Sigma}$ and that $\hat{\Sigma}$ generates $\Z^2$.
By using this sequence, to show that $\check{A} - 1 \in \cCs = \Ker (\sigma \circ \check{\gamma})$, it is sufficient to show $(\sigma \circ \check{\gamma}) (\check{A} - 1) = 0$.
This holds since $(\sigma \circ \check{\gamma}) (\check{\cP}_{m,n}) = 1$ for any $(m,n) \in \Z^2$.
\end{proof}

By using Theorem~\ref{construction}, we here compute $K$-groups of concave corner $C^*$-algebras $\TTs$ and its commutator ideals $\cCs$.
Associated with the sequence (\ref{main}), we have the following six-term exact sequence:
\[\xymatrix{
K_1(K(\HHs)) \ar[r]& K_1(\TTs) \ar[r] & K_1(\Sab) \ar[d]^{\check{\delta}_1}
\\
K_0(\Sab) \ar[u]^{\check{\delta}_0} & K_0(\TTs) \ar[l] & K_0(K(\HHs)), \ar[l]
}\]
$K$-groups of $\Sab$ is computed in \cite{Pa90}.
The result is\footnote{In what follows, $K$-groups of $C^*$-algebras $\TTs$ and $\cCs$ are computed, and the result for $\Sab$, $\cCa$ and $\cCb$ are presented corresponding to the values of $\alpha$ and $\beta$. The case of $\alpha = -\infty$ or $\beta = + \infty$ is the same as that of rational $\alpha$ or rational $\beta$.},
$$
K_0(\Sab) \cong
\left\{
\begin{aligned}
\Z \ & \hspace{3mm} \text{if $\alpha$ and $\beta$ are both rational,}\\
\Z^2	& \hspace{3mm} \text{if either $\alpha$ or $\beta$ is rational and the other is irrational,}\\
\Z^3	& \hspace{3mm} \text{if $\alpha$ and $\beta$ are both irrational.}
\end{aligned}
\right.
$$
and $K_1(\Sab) = \Z$.
By Theorem~\ref{construction}, we can see from the above six-term exact sequence that $\check{\delta}_1 \colon K_1(\Sab) \to K_0(K(\HHs))$ is an isomorphism.
Now, we can compute $K$-groups of $\TTs$, and the result is as follows:
$$
K_0(\TTs) \cong
\left\{
\begin{aligned}
\Z \ & \hspace{3mm} \text{if $\alpha$ and $\beta$ are both rational,}\\
\Z^2	& \hspace{3mm} \text{if either $\alpha$ or $\beta$ is rational and the other is irrational,}\\
\Z^3	& \hspace{3mm} \text{if $\alpha$ and $\beta$ are both irrational.}
\end{aligned}
\right.
$$
and $K_1(\TTs) = 0$.

We next compute the $K$-group of the commutator algebra $\cCs$.
As in \cite{Ji95,Pa90}, if we restrict the sequence (\ref{exact}) to $\cCs \subset \TTs$, we obtain the following short exact sequence:
\begin{equation*}
	0 \to K(\HHs) \to \cCs \overset{\check{\gamma}'}{\to} \cCa \oplus \cCb \to 0,
\end{equation*}
where $\check{\gamma}'$ is the restriction of $\check{\gamma}$ onto $\cCs$.
Associated with this sequence, we have the following six-term exact sequence:
\[\xymatrix{
K_1(K(\HHs)) \ar[r]& K_1(\cCs) \ar[r] & K_1(\cCa) \oplus K_1(\cCb) \ar[d]
\\
K_0(\cCa) \oplus K_0(\cCb) \ar[u]& K_0(\cCs) \ar[l] & K_0(K(\HHs)). \ar[l]
}\]
$K$-groups of $\cCa$ and $\cCb$ are computed in \cite{JK88,Xia88}.
The result is
$$
K_0(\cCa) \cong K_0(\cCb) \cong
\left\{
\begin{aligned}
\Z \ & \hspace{3mm} \text{if $\alpha$ (or $\beta$) is rational},\\
\Z^2	& \hspace{3mm} \text{if $\alpha$ (or $\beta$) is irrational}.
\end{aligned}
\right.
$$
and $K_1(\cCa) \cong K_1(\cCb) \cong \Z$.
By Theorem~\ref{construction}, the operator $\check{A} - 1$ is contained in the algebra $\cCs$. Thus, the map $\check{\delta}_1 \colon K_1(\cCa) \oplus K_1(\cCb) \to K_0(K(\HHs))$ is surjective.
As in \cite{Ji95}, we can compute $K$-groups of $\TTs$ as follows:
$$
K_0(\cCs) \cong
\left\{
\begin{aligned}
\Z^2	& \hspace{3mm} \text{if $\alpha$ and $\beta$ are both rational},\\
\Z^3	& \hspace{3mm} \text{if either $\alpha$ or $\beta$ is rational, and the other is irrational},\\
\Z^4	& \hspace{3mm} \text{if $\alpha$ and $\beta$ are both irrational}.
\end{aligned}
\right.
$$
and $K_1(\cCs) = \Z$.

\subsection{A relation with the quarter-plane case and an index formula}
We next compare index theory for quarter-plane (convex corner) Toeplitz operators \cite{DH71,Ji95,Pa90} and that for concave corner Toeplitz operators.
There are group isomorphisms $K_0(\hat{\mathrm{Tr}}) \colon K_0(K(\HHab)) \to \Z$ and $K_0(\check{\mathrm{Tr}}) \colon K_0(K(\HHs))$ $\to \Z$ that map a class $[p]_0$ of a finite-rank projection $p$ to $\rank(\mathrm{Image} (p))$.
Associated with extensions (\ref{seq1}) and (\ref{exact}), there are the following two group isomorphisms from $K_1(\Sab)$ to $\Z$:
\begin{equation*}
	 K_0(\hat{\mathrm{Tr}}) \circ \hat{\delta}_1 \colon K_1(\Sab) \to \Z \ \ \ \text{and} \ \ \ K_0(\check{\mathrm{Tr}}) \circ \check{\delta}_1 \colon K_1(\Sab) \to \Z.
\end{equation*}
where $\hat{\delta}_1$ is the boundary homomorphism of the six-term exact sequence of $K$-theory for $C^*$-algebras associated with the quarter-plane Toeplitz extension (\ref{seq1}).
Let $\hat{\cP}_{m,n} := \Pab M_{m,n} \Pab M_{-m,-n} \Pab$.
Jiang considered in \cite{Ji95} the following operator:
\begin{equation*}
	\hat{A} := \hat{\cP}_{0,1} + M_{1,1}(1 - \hat{\cP}_{-1,0}) + M_{1,0}(\hat{\cP}_{-1,0} - \hat{\cP}_{0,1}).
\end{equation*}
Note that $\hat{A} \in \TTab$.
It is shown in \cite{Ji95} that $\hat{A}$ is the isometric Fredholm operator whose Fredholm index is $-1$.
By comparing the Fredholm quarter-plane Toeplitz operator constructed in \cite{Ji95} and the Fredholm concave corner Toeplitz operator constructed in Theorem~\ref{construction}, we obtain the following result:
\begin{corollary}\label{relation}
$K_0(\check{\mathrm{Tr}}) \circ \check{\delta}_1= -K_0(\hat{\mathrm{Tr}}) \circ \hat{\delta}_1$.
\end{corollary}
\begin{proof}
By the map $\hat{\gamma}$ in the quarter-plane Toeplitz extension (\ref{seq1}), we have $\hat{\gamma}(\hat{A}) \in \Sab$.
We can check the equality $\hat{\gamma}(\hat{A}) = \check{\gamma}(\check{A})$ and that this element gives a generator $[\hat{\gamma}(\hat{A})]_1$ of the $K$-group $K_1(\Sab) \cong \Z$.
By Theorem~$1$ of \cite{Ji95}, Theorem~\ref{construction} and Proposition~$9.4.2$ of \cite{RLL00}, we have 
$(K_0(\check{\mathrm{Tr}}) \circ \check{\delta}_1)([\check{\gamma}(\check{A})]_1) = \ind(\check{A}) = 1 = - \ind(\hat{A}) = - (K_0(\hat{\mathrm{Tr}}) \circ \hat{\delta}_1)([\hat{\gamma}(\hat{A})]_1)$.
\end{proof}

We now restrict our attention to the case of $\alpha = 0$ and $\beta = \infty$.
Previous works studied quarter-plane Toeplitz operators in this case and obtained many results \cite{CDS72,DH71}.
Combined with Corollary~\ref{relation}, we obtain corresponding results for concave corner Toeplitz operators, and we state it explicitly for the later use.

In \cite{CDS72}, Coburn--Douglas--Singer obtained an index formula for Fredholm quarter-plane Toeplitz operators.
A corresponding result for concave corner Toeplitz operators is as follows.
Let $r$ be a positive integer.
The map $\hat{\gamma}$ induces a surjective $*$-homomorphism
\begin{equation*}
1 \otimes \hat{\gamma} \colon M_r(\C)  {\otimes} \check{\TT}^{0,\infty} \cong M_r(\check{\TT}^{0,\infty}) \to M_r(\C)  {\otimes} \mathcal{S}^{0,\infty} \cong M_r(\mathcal{S}^{0,\infty})
\end{equation*}
which we denote $\hat{\gamma}$, for simplicity\footnote{If $A$ is an algebra, $M_r(A)$ denotes the algebra of all $r \times r$ matrices with entries in $A$.}.
The algebra $M_r(\mathcal{S}^{0,\infty})$ is a $C^*$-subalgebra of $M_r(\mathcal{T}^0) \oplus M_r(\mathcal{T}^\infty) \cong M_r(C(\T)  {\otimes} \mathcal{T}) \oplus M_r(\mathcal{T} {\otimes} C(\T))$.
We write $(\xi, \eta)$ for valuables in $\T^2$.
$M_r(\mathcal{T}^0)$ and $M_r(\mathcal{T}^\infty)$ have valuables $\xi$ and $\eta$, respectively.

\begin{corollary}\label{concaveCDS}
If $\check{T}$ is a Fredholm operator in $M_r(\check{\TT}^{0,\infty})$ with symbol $\check{\gamma}(\check{T}) = (\check{\gamma}^0(\check{T}), \check{\gamma}^\infty(\check{T}))$ in $M_r(\mathcal{S}^{0,\infty})$.
Then, there is a path $(F_t, G_t)$ in $M_r(\mathcal{S}^{0,\infty})$ such that $F_0 = \check{\gamma}^0(\check{T})$, $G_0 =  \check{\gamma}^\infty(\check{T})$ and such that
\begin{equation*}
F_1(\xi) =
\begin{pmatrix}
\xi^m & & & \\
 & 1 &  & \\
 & 	       & \ddots & \\
 &  & & 1
\end{pmatrix}
\ \text{and} \ \
G_1(\eta) =
\begin{pmatrix}
\eta^n & & & \\
 & 1 &  & \\
 & 	       & \ddots & \\
 &  & & 1
\end{pmatrix}
\end{equation*}
for some $(m,n)$ in $\Z^2$.
The Fredholm index of $\check{T}$ is given by $\ind(\check{T}) = m+n$.
\end{corollary}

\begin{proof}
Since $\hat{\gamma} \colon M_r(\hat{\TT}^{0,\infty}) \to M_r(\mathcal{S}^{0,\infty})$ is surjective, there is $\hat{T} \in M_r(\hat{\TT}^{0,\infty})$ satisfying $\hat{\gamma}(\hat{T}) = \check{\gamma}(\check{T})$.
Since $\check{T}$ is Fredholm, $\check{\gamma}(\check{T})$ is invertible in $\mathcal{S}^{0,\infty}$, and thus, $\hat{T}$ is a Fredholm quarter-plane Toeplitz operator (see Theorem $2.6$ of \cite{Pa90}).
By Theorem in p$589$ of \cite{CDS72}, such a path $(F_t, G_t)$ exists, and we have $\ind(\hat{T}) = -(m+n)$.
By Corollary~\ref{relation}, we have  $-(m+n) = \ind(\hat{T}) = (K_0(\hat{\mathrm{Tr}}) \circ \hat{\delta}_1) ([\hat{\gamma}(\hat{T})]_1) = -K_0(\hat{\mathrm{Tr}}) \circ \hat{\delta}_1 ([\check{\gamma}(\check{T})]_1) = - \ind(\check{T})$.
\end{proof}

\begin{remark}\label{notunique}
According to \cite{CDS72}, the path $(F_t, G_t)$ is not unique and each $m$ and $n$ are not uniquely determined in general.
\end{remark}

We next see that when a Fredholm concave corner Toeplitz operators is of a special form, its Fredholm index is zero.
The corresponding result for quarter-plane Toeplitz operators is obtained in \cite{DH71}.
For a continuous function, $\varphi \colon \T^2 \to \C$, the multiplication operator generated by $\varphi$ defines a bounded linear operator on $L^2(\T^2)$.
Through the Fourier transform, this multiplication operator defines a bounded linear operator $M_\varphi$ on $l^2(\Z^2)$.
Then, $\check{P}^{0,\infty} M_\varphi \check{P}^{0,\infty}$ is a concave corner Toeplitz operator.
For an operator of this form, we have the following result.

\begin{corollary}\label{zero}
Let $\varphi \colon \T^2 \to \C$ be a continuous function.
If the concave corner Toeplitz operator $\check{P}^{0,\infty} M_\varphi \check{P}^{0,\infty}$ is Fredholm, then its Fredholm index is zero.
\end{corollary}

\begin{proof}
By our assumption, $\check{\gamma}(\check{P}^{0,\infty} M_\varphi \check{P}^{0,\infty}) = (P^{0} M_\varphi P^{0}, P^{\infty} M_\varphi P^{\infty})$ is an invertible element in $\mathcal{S}^{0,\infty}$.
Thus, $\hat{\gamma}(\hat{P}^{0,\infty} M_\varphi \hat{P}^{0,\infty}) = (P^{0} M_\varphi P^{0}, P^{\infty} M_\varphi P^{\infty})$ also is invertible, and the quarter-plane Toeplitz operator $\hat{P}^{0,\infty} M_\varphi \hat{P}^{0,\infty}$ is Fredholm.
By Corollary in p$208$ of \cite{DH71}, the Fredholm index of $\hat{P}^{0,\infty} M_\varphi \hat{P}^{0,\infty}$ is zero.
By Corollary~\ref{relation}, the Fredholm index of $\check{P}^{0,\infty} M_\varphi \check{P}^{0,\infty}$ is also zero.
\end{proof}

\section{Topological invariants and topologically protected corner states}
In \cite{Hayashi2}, 3-D class A systems with codimension-two convex corners are studied, and a topological invariant is defined for a gapped bulk-edges Hamiltonian as an element of some $K$-group.
Its relation with gapless corner states is also proved.
Key ingredients are index theory for quarter-plane Toeplitz operators \cite{DH71,Ji95,Pa90}.
A nontrivial example is obtained in \cite{Hayashi2} by constructing Hamiltonians from Hamiltonians of 2-D class A and 1-D class AIII (conventional) topological insulators.
For such Hamiltonians, if we consider the convex corner of the special shape ($\alpha = 0$ and $\beta = +\infty$), corner topological invariants are defined, and the numerical corner invariant is equal to the product of two topological numbers of two (conventional) topological insulators (called the product formula).
The study in \cite{Hayashi2} is based on previous works \cite{DH71,Ji95,Pa90} and is restricted to convex corners.

The results in Sect.~$2$ and Sect.~$3$ of this paper enable us to examine systems with concave corners.
In this section, we define topological invariants for some Hamiltonians on 2-D class AIII systems (Sect.~$4.1$).
Moreover, we introduce the gapless corner topological invariant especially for concave corners and show the relation between gapped and gapless invariants. 
Correspondingly, for Hamiltonians that are gapped on two edges, we can define two corner invariants corresponding to these two corners.
We show that there is a relation between these two corner invariants.
We further formulate a product formula as in \cite{Hayashi2}.
By using this, we obtain explicit examples of gapped bulk-edges Hamiltonians of nontrivial convex and concave corner invariants.
They differ by the multiplication of $-1$,
which clarifies that these corner invariants depend on the shape of the system.
Since 3-D class A systems with convex corners are studied in \cite{Hayashi2}, we mainly consider other cases.

We here collect the notations used in this subsection.
Let $V$ be a finite rank Hermitian vector space and denote the complex dimension of $V$ by $N$.
We write $\HH_V$, $\HHa_V$, $\HHb_V$, $\HHab_V$ and $\HHs_V$ for $\HH  {\otimes} V$, $\HHa  {\otimes} V$, $\HHb  {\otimes} V$, $\HHab  {\otimes} V$ and $\HHs  {\otimes} V$, respectively.
If there is an endomorphism $\Pi$ on $V$, we extend $\Pi$ onto $\HH_V$, $\HHa_V$, $\HHb_V$, $\HHab_V$ and $\HHs_V$ by the pointwise operation, i.e., $1  {\otimes} \Pi$, and denote $\Pi$, also.
Similarly, we write $\Pa$, $\Pb$ $\Pab$ and $\Ps$ for the orthogonal projections onto $\HHa_V$, $\HHb_V$, $\HHab_V$ and $\HHs_V$ defined by $\Pa  {\otimes} 1$, $\Pb  {\otimes} 1$, $\Pab  {\otimes} 1$ and $\Ps  {\otimes} 1$, for simplicity.

\subsection{2-D class AIII system}
In this subsection, we consider 2-D class AIII systems with a codimension-two corner.
We rather focus on the study of systems with concave corners,
but we briefly study convex corners and show a relation between corner invariants defined on systems with these two types of corners.

In this subsection, we assume that the vector space $V$ has a $\Z_2$-grading given by $\Pi$.
Specifically, $\Pi \colon V \to V$ is a complex linear map that satisfies $\Pi^2 = 1$.
We consider a continuous map $\T^2 \to \End_\C(V)$, $(\xi, \eta) \mapsto H(\xi, \eta)$, where, for each $(\xi,\eta) \in \T^2$, $H(\xi,\eta)$ is Hermitian.
Moreover, we assume that $H(\xi, \eta)$ preserves chiral symmetry, that is, for any $(\xi,\eta) \in \T^2$, $H(\xi,\eta)$ anti-commutes with $\Pi$.
Note that in this case, $N$ is necessarily an even number.
Through the Fourier transform, the multiplication operator on $L^2(\T^2; V)$ generated by $H(\xi, \eta)$ defines a bounded linear self-adjoint operator $H$ on the Hilbert space $\HH_V$.
We call $H$ the {\em bulk Hamiltonian}.
Let $\alpha < \beta$ be real numbers (possibly $\alpha = -\infty$ or $\beta = +\infty$, but not both).
By using them, we consider the following half-plane Toeplitz operators,
\begin{equation*}
	H^\alpha := \Pa H \Pa \colon \HHa_V \to \HHa_V, \ \ H^\beta := \Pb H \Pb \colon \HHb_V \to \HHb_V.
\end{equation*}
and call them {\em edge Hamiltonians}.
We also consider the following convex and concave corner Toeplitz operators:
\begin{equation*}
	\hat{H}^{\alpha,\beta} := \Pab H \Pab \colon \HHab_V \to \HHab_V, \ \ \check{H}^{\alpha,\beta} := \Ps H \Ps \colon \HHs_V \to \HHs_V,
\end{equation*}
and call them {\em corner Hamiltonians}.

Note that $H^\alpha$, $H^\beta$, $\hat{H}^{\alpha,\beta}$ and $\check{H}^{\alpha,\beta}$ anti-commutes with $\Pi$.
The following is our assumption in this subsection.

\vspace{1mm}
\noindent
{\bf Assumption (Spectral gap condition)}
We assume that our edge Hamiltonians have a common spectral gap at the Fermi level $0$, i.e.,
$0$ is not contained in either $\mathrm{sp}(H^\alpha)$ or $\mathrm{sp}(H^\beta)$.
We refer to this condition as the {\em spectral gap condition}.
\vspace{1mm}

\noindent
Note that under this assumption, our bulk Hamiltonian is also gapped at zero \cite{Hayashi2}. 
By using chiral symmetry, we have following decomposition:
$H =
\begin{pmatrix}
0 & h^*\\
h & 0
\end{pmatrix}$.
We now fix an orthonormal basis of $V$ and identify it with $\C^N$.
By our spectral gap condition, the operators $\Pa h \Pa$ and $\Pb h \Pb$ are both invertible.
Let $u^\alpha := \Pa h \Pa /|\Pa h \Pa|$ and $u^\beta = \Pb h \Pb/|\Pb h \Pb|$.\footnote{$\Pa h \Pa /|\Pa h \Pa|$ is defined by the continuous functional calculous by the continuous function $\C \setminus \{ 0\} \to \C$ given by $z \mapsto z/|z|$.}
The pair $(u^\alpha, u^\beta)$ defines a unitary element in $M_{N/2}(\Sab)$ and so defines an element of the $K_1$-group $K_1(\Sab)$.\footnote{This element does not depend on the choice of the identification $V \cong \C^N$.}
\begin{definition}\label{gappedinvAIII}
We define the gapped topological invariant as follows:
\begin{equation*}
	\I_\BE^{2d, \AIII}(H) :=  [(u^\alpha, u^\beta)]_1 \in K_1(\Sab).
\end{equation*}
\end{definition}
We next consider a system with a concave corner and introduce a corner topological invariant.
By the spectral gap condition, $\check{\gamma}(\check{H}^{\alpha,\beta}) = (H^\alpha, H^\beta)$ and $\check{\gamma}(\Ps h \Ps) = (\Pa h \Pa, \Pb h \Pb)$ are invertible elements.
Thus, by Theorem~\ref{Fredholm}, the operators $\check{H}^{\alpha,\beta}$ and $\Ps h \Ps$ are Fredholm.
By the polar decomposition, there is a unique partial isometry $v \in M_{N/2}(\TTs)$ such that $\Ps h \Ps = v |\Ps h \Ps|$.
By using this, we define the corner invariant.
\begin{definition}\label{gaplessinvAIII}
We define the {\em gapless corner invariant} of our system as follows:
\begin{equation*}
	\check{\I}_\Corner^{2d, \AIII}(H) := [1-v^*v]_0 - [1- vv^*]_0 \in K_0(K(\HHs)).
\end{equation*}
\end{definition}

The following is the bulk-edge and corner correspondence for our system.
\begin{theorem}\label{BECCAIII}
The map $\check{\delta}_1 \colon K_1(\Sab) \to K_0(K(\HHs))$ maps the gapped topological invariant $\I_\BE^{2d, \AIII}(H)$ to the gapless corner invariant $\check{\I}_\Corner^{2d, \AIII}(H)$.
\end{theorem}
\begin{proof}
Since $\hat{\gamma}(v) = (u^\alpha, u^\beta)$, this follows from Proposition~$9.2.4$ of \cite{RLL00}.
\end{proof}

By using the isomorphism $K_0(\check{\mathrm{Tr}}) \colon K_0(K(\HHs)) \to \Z$, we obtain an integer, i.e., the {\em numerical} corner invariant.
We here write it down explicitly.
Since $\check{H}^{\alpha,\beta}$ is Fredholm, $\Ker \check{H}^{\alpha,\beta}$ is of finite rank.
Since $\Pi$ anti-commutes with $\check{H}^{\alpha,\beta}$, $\Pi$ acts on $\Ker \check{H}^{\alpha,\beta}$.
Moreover, since $\Pi^2 = 1$, the space $\Ker \check{H}^{\alpha,\beta}$ decomposes into the direct sum of $+1$ eigenspace $W^+$ and $-1$ eigenspace $W^-$ of $\Pi |_{\Ker \check{H}^{\alpha,\beta}}$.
We define its {\em signature} $\mathrm{sign}( \Pi |_{\Ker \check{H}^{\alpha,\beta}})$ as the difference of the rank of these spaces, that is,
\begin{equation*}
	\mathrm{sign}(\Pi |_{\Ker \check{H}^{\alpha,\beta}}) := \rank W^+ - \rank W^-.
\end{equation*}
Note that the signature is used to define edge indices of 1-D class AIII topological insulators (see \cite{PS16}, for example).
By using this, the {\em numerical corner invariant} of our 2-D class AIII system with a codimension-two concave corner is expressed as follows:
\begin{equation*}
	K_0(\check{\mathrm{Tr}})(\check{\I}_\Corner^{2d, \AIII}(H)) =  \ind(\Ps h \Ps) = \mathrm{sign}( \Pi |_{\Ker \check{H}^{\alpha,\beta}}) \in \Z.
\end{equation*}

By using the extension (\ref{seq1}) instead of (\ref{exact}), we can also treat convex corners in the same way.
By using the convex corner Hamiltonian $\hat{H}^{\alpha,\beta}$, the corner topological invariant $\hat{\I}_\Corner^{2d, \AIII}(H)$ for a 2-D class AIII system with codimension-two convex corner is defined as an element of the $K$-group $K_0(K(\HHab))$.
Its numerical corner invariant is given by $K_0(\hat{\mathrm{Tr}})(\hat{\I}_\Corner^{2d, \AIII}(H)) = \mathrm{sign}( \Pi |_{\Ker \hat{H}^{\alpha,\beta}})$.
Moreover, the bulk-edge and corner correspondence holds; that is, 
\begin{equation}\label{BECCAIII2}
	\hat{\delta}_1(\I_\BE^{2d, \AIII}(H)) = \hat{\I}_\Corner^{2d, \AIII}(H).
\end{equation}

The following is a relation between numerical corner invariants for convex and concave corners.
\begin{theorem}\label{minusAIII}
$K_0(\check{\mathrm{Tr}})(\check{\I}_\Corner^{2d, \AIII}(H)) = - K_0(\hat{\mathrm{Tr}})(\hat{\I}_\Corner^{2d, \AIII}(H))$
\end{theorem}
\begin{proof}
This follows from Corollary~\ref{relation}, Theorem~\ref{BECCAIII} and (\ref{BECCAIII2}).
\end{proof}

We now compare our gapped topological invariant $\I_\BE^{2d, \AIII}(H)$ with {\em bulk} topological invariants for 2-D class AIII topological insulators.
Let $u := h/|h|$.
Through the Fourier transform, $u$ defines a unitary element in $M_{N/2}(C(\T^2))$ and thus defines an element $[u]_1$ in $K_1(C(\T^2))$.
We have $K_1(C(\T^2)) \cong \Z \oplus \Z$ and topological invariants for the bulk Hamiltonian corresponding to these two $\Z$ components are called {\em weak invariants}.
\begin{proposition}\label{weakAIII}
	For Hamiltonians satisfying our spectral gap condition, these two weak invariants are zero.
\end{proposition}
\begin{proof}
The algebra $\Sab$ is defined as a pullback.
By calculating the Mayer-Vietoris exact sequence for the pull-back diagram (\ref{Sab}), we can check that $\sigma_* \colon K_1(\Sab)$ $\to K_1(C(\T^2))$ is the zero map.
Since $\sigma_*([(u^\alpha, u^\beta)]_1) = [u]_1$, we have $[u]_1 = 0$, which means that these two weak invariants are both zero.
\end{proof}

We next restrict our attention to the case of $\alpha = 0$ and $\beta = \infty$ and consider an explicit example.
We first see the following constraint.
\begin{proposition}\label{remrank}
When $N=\rank V$ is $2$, the corner invariants $\hat{\I}_\Corner^{2d, \AIII}(H)$ and $\check{\I}_\Corner^{2d, \AIII}(H)$ for convex and concave corners are both zero.
\end{proposition}
\begin{proof}
We first consider the case of convex corners.
Since $K_0(\hat{\mathrm{Tr}})$ is an isomorphism, it is sufficient to show that $K_0(\hat{\mathrm{Tr}})(\hat{\I}_\Corner^{2d, \AIII}(H))$ is zero.
Note that $K_0(\hat{\mathrm{Tr}})(\hat{\I}_\Corner^{2d, \AIII}(H)) = \ind(\hat{P}^{0,\infty} h \hat{P}^{0,\infty})$.
When $N = 2$, $h$ is a Fourier transform of a multiplication operator on $L^2(\T^2)$ generated by a continuous function $\T^2 \to \C$.
Then, the results follow from Corollary in p$208$ of \cite{DH71}.
The result for concave cases follows from Corollary~\ref{zero}.
\end{proof}
Thus, to find 2-D class AIII Hamiltonians of nontrivial corner invariants, $N$ must be greater than or equal to $4$ since $N$ is an even integer.

We now give a construction of nontrivial examples.
For $j=1,2$, let $V_j$ be $\Z_2$-graded finite rank Hermitian vector spaces whose $\Z_2$-gradings are given by complex linear maps $\Pi_j \colon V_j \to V_j$ that satisfy $\Pi_j^2 = 1$ $(j=1,2)$.
Let $H_j$ be multiplication operators on $l^2(\Z; V_j)$ generated by continuous maps $\T \to \End(V_j)$, $t \mapsto H_j(t)$.
We assume that $H_j$ is self-adjoint invertible and satisfies the relation $\Pi_j H_j \Pi_j^* = -H_j$ $(j=1,2)$.
$H_1$ and $H_2$ are Hamiltonians of 1-D class AIII (conventional) topological insulators.
Let $\I^{1d, \AIII}(H_1)$ and $\I^{1d, \AIII}(H_2)$ be their topological invariants, which are defined as follows\footnote{We here give the definition of edge topological invariants for 1-D class AIII topological insulators. By the bulk-edge correspondence, this coincides with the bulk topological invariant which is defined as the winding number of the determinant of its symbol, that is $\mathrm{Wind}(\{\det h_j(t)\}_{t \in \T})$ where $H_j =
\begin{pmatrix}
0 & h_j^*\\
h_j & 0
\end{pmatrix}$. (see \cite{PS16}, for example).}.
Let $\Ker H_1 = W_1^+ \oplus W_1^-$ be the eigenspace decomposition with respect to $\Pi_1$, where the action of $\Pi_1$ on $W_1^\pm$ is $\pm 1$, respectively.
Let $w_1^+ = \rank W_1^+$ and $w_1^- = \rank W_1^-$. Then, we have $\I^{1d, \AIII}(H_1) = -w_1^+ + w_1^-$.
We also take the eigenspace decomposition $\Ker H_2 = W_2^+ \oplus W_2^-$ with respect to $\Pi_2$ and let $w_2^+ = \rank W_2^+$ and $w_2^- = \rank W_2^-$.
Then, we have $\I^{1d, \AIII}(H_2) = -w_2^+ + w_2^-$.
We consider the following operator on $l^2(\Z^2; V_1  {\otimes} V_2)$:
\begin{equation*}
	H = H_1  {\otimes} \Pi_2 + 1  {\otimes} H_2.
\end{equation*}
which has a chiral symmetry given by $\Pi := \Pi_1  {\otimes} \Pi_2$.
Then, the bulk and two edge Hamiltonians $H$, $H^0$ and $H^\infty$ are all invertible, i.e., gapped at zero (see Theorem $4$  (1) of \cite{Hayashi2}).
Moreover, the following formulae hold:
\begin{theorem}\label{prodthmAIII}
\begin{enumerate}
\renewcommand{\labelenumi}{(\arabic{enumi})}
	\item $K_0(\hat{\mathrm{Tr}})(\hat{\I}_\Corner^{2d,\AIII}(H)) = \I^{1d, \AIII}(H_1) \cdot \I^{1d, \AIII}(H_2)$,
	\item $K_0(\check{\mathrm{Tr}})(\check{\I}_\Corner^{2d,\AIII}(H)) = - \I^{1d, \AIII}(H_1) \cdot \I^{1d, \AIII}(H_2)$.
\end{enumerate}
\end{theorem}
\begin{proof}
As in Theorem~$4$ of \cite{Hayashi2}, $\Ker \hat{H}^{0,\infty} = \Ker H_1  {\otimes} \Ker H_2$ holds.
We have
\begin{eqnarray*}
	\Ker \hat{H}^{0,\infty} 	&=&	(W_1^+ \oplus W_1^-)  {\otimes} (W_2^+ \oplus W_2^-)\\
			&=& (W_1^+  {\otimes} W_2^+) \oplus (W_1^-  {\otimes} W_2^+) \oplus (W_1^+  {\otimes} W_2^-) \oplus (W_1^-  {\otimes} W_2^-).
\end{eqnarray*}
The operator $\Pi$ acts on this space, and we have
\begin{align*}
	&K_0(\hat{\mathrm{Tr}})(\hat{\I}_\Corner^{2d,\AIII}(H)) = \sign \Pi|_{\Ker \hat{H}^{0,\infty}}\\
		&=	\rank(W_1^+  {\otimes} W_2^+) - \rank(W_1^-  {\otimes} W_2^+) - \rank(W_1^+  {\otimes} W_2^-) + \rank(W_1^-  {\otimes} W_2^-)\\
		&= w_1^+ w_2^+ - w_1^- w_2^+ - w_1^+ w_2^- + w_1^-w_2^-
		= (w_1^+ - w_1^-)(w_2^+ - w_2^-)\\
		&= \I^{1d, \AIII}(H_1) \cdot \I^{1d, \AIII}(H_2).
\end{align*}
This proves (1).
(2) follows from (1) and Theorem~\ref{minusAIII}.
\end{proof}
Note that to find $H_1$ and $H_2$ of nontrivial topological invariants, the rank of $V_1$ and $V_2$ must be greater than or equal to $2$.
Thus, to find an example of a nontrivial corner invariant in this way, the rank of $V_1  {\otimes} V_2$ must be greater than or equal to $4$.
This is consistent with Proposition~\ref{remrank}, and an example contained in Sect.~\ref{sectexam} provides an example of $N=4$.

\begin{remark}\label{CDScorner}
Numerical corner invariants for convex and concave corners are given by Fredholm indices of convex and concave corner Toeplitz operators, respectively.
When $\alpha = 0$ and $\beta =\infty$, the Coburn--Douglas--Singer index formula \cite{CDS72} and its concave analog (Corollary~\ref{concaveCDS}) give a topological method to compute them by using gapped Hamiltonians.
However, to find a necessary path in the algebra $M_{N/2}(\mathcal{S}^{0,\infty})$ is not necessarily easy in general \cite{CDS72,Pa90}.
\end{remark}

\begin{remark}
Since we defined topological invariants ($\I_\BE^{2d, \AIII}(H)$ in Definition~\ref{gappedinvAIII} and $\check{\I}_\Corner^{2d, \AIII}(H)$ in Definition~\ref{gaplessinvAIII}) and stated their relation (Theorem~\ref{BECCAIII}) in a $K$-theoretic way, a generalization to the higher-dimensional case is straightforward, as in Remark~$5$ of \cite{Hayashi2}.
For a $(n+2)$-D class AIII system with codimension-two concave corner, a topological invariant for gapped bulk-edges Hamiltonians is defined as an element of $K_1(\Sab  {\otimes} \T^{n})$, and a gapless corner invariant is defined as that of $K_0(K(\HHs)  {\otimes} \T^{n})$.
Let $\check{\delta}_1 \colon K_1(\Sab  {\otimes} \T^{n})$ $\to$ $K_0(K(\HHs)  {\otimes} \T^{n})$ be a boundary homomorphism associated with a short exact sequence obtained by taking a tensor product of the sequence (\ref{exact}) and $C(\T^n)$.
Then, $\check{\delta}_1$ maps the gapped topological invariant to the gapless corner invariant.
Its definition and proof are parallel with the one in this subsection.
\end{remark}

\subsection{3-D class A system}
In this subsection, we consider 3-D class A systems with codimension-two concave corners.
The contents of this section are almost parallel with \cite{Hayashi2}, but we here use the sequence (\ref{exact}) instead of the quarter-plane Toeplitz extension (\ref{seq1}) used in \cite{Hayashi2}.

We consider a continuous map $\T^3 \to \End_\C(V)$, $(\xi, \eta, t) \mapsto H(\xi, \eta, t)$, where, for each $(\xi,\eta,t) \in \T^3$, $H(\xi,\eta,t)$ is Hermitian.
The multiplication operator generated by $H(\xi, \eta, t)$ defines a bounded linear operator on $L^2(\T^3;V)$.
Through the Fourier transform, we obtain a bounded linear self-adjoint operator $H$ on $l^2(\Z^3;V)$ and
We call $H$ the {\em bulk Hamiltonian}.
By the Fourier transform in the last $\Z$ component, we obtain a continuous family of bounded linear self-adjoint operators $\{ H(t) \colon \HH_V \to \HH_V \}_{t \in \T}$.
By taking their compressions onto $\HHa_V$ and $\HHb_V$, we obtain one-parameter families of half-plane Toeplitz operators,
 \begin{equation*}
	\{ H^\alpha(t) := \Pa H(t) \Pa \}_{t \in \T}, \ \
	\{ H^\beta(t) := \Pb H(t) \Pb \}_{t \in \T},
\end{equation*}
and we call them {\em edge Hamiltonians}.
We also consider the compression onto $\HHs_V$ and obtain the following family of concave corner Toeplitz operators:
\begin{equation*}
	\{ \check{H}^{\alpha,\beta}(t) := \Ps H(t) \Ps\}_{t \in \T},
\end{equation*}
We call them the {\em corner Hamiltonian}.
The following is our assumption in this subsection.

\vspace{1mm}
\noindent
{\bf Assumption (Spectral gap condition)}
We assume that our edge Hamiltonians have a common spectral gap at the Fermi level $\mu \in \R$ for any $t$ in $\T$, i.e.,
$\mu$ is not contained in either $\mathrm{sp}(H^\alpha(t))$ or $\mathrm{sp}(H^\beta(t))$.
We refer to this condition as a {\em spectral gap condition}.
\vspace{1mm}

\noindent
In what follows, we assume $\mu = 0$ without loss of generality.
Under the spectral gap condition, the gapped topological invariant is defined as an element of a $K$-group, that is,
$\I_\BE^{3d, \A}(H) \in K_0(\mathcal{S}^{\alpha,\beta}  {\otimes} C(\T))$ (defined at Definition~$1$ of \cite{Hayashi2} and denoted $\I_\BE(H)$ there).
We here consider a concave corner that appears as a union of two half-planes and defines the corner invariant.

\begin{definition}\label{concavecornerinv}
By the spectral gap condition and Theorem~\ref{main}, we have a continuous family $\{ \check{H}^{\alpha,\beta}(t) \}_{t \in \T}$ of bounded linear self-adjoint Fredholm operators.
This family defines an element $\check{\I}_\Corner^{3d, \A}(H)$ of the $K$-group $K_1(C(\T))$.
We call $\check{\I}_\Corner^{3d, \A}(H)$ the {\em gapless corner invariant}.
\end{definition}
Its numerical corner invariant is given by using spectral flow\footnote{We here regard $\T$ as the unit circle in the complex plane and fix the counter-clockwise orientation.} $\mathrm{sf} \colon$ $K_1(C(\T)) \to \Z$, that is, $\mathrm{sf}(\check{\I}_\Corner^{3d, \A}(H)) \in \Z$.
The following is the bulk-edge and corner correspondence for our system.
\begin{theorem}\label{BECCA}
The map $\check{\delta_0} \colon K_0(\Sab  {\otimes} C(\T)) \to K_1(C(\T))$ maps $\I_\BE^{3d, \A}(H)$ to the gapless corner invariant.
That is,
$\check{\delta}_0(\I_\BE^{3d, \A}(H)) =\check{\I}_\Corner^{3d, \A}(H)$.
\end{theorem}
Definition~\ref{concavecornerinv} and Theorem~\ref{BECCA} are parallel with Definition~$2$ and Theorem~$3$ of \cite{Hayashi2}, and we omit the detail.
In our setting, we can define convex and concave corner invariants $\hat{\I}_\Corner^{3d, \A}(H)$ and $\check{\I}_\Corner^{3d, \A}(H)$ for convex and concave corners, respectively (the convex corner invariant $\hat{\I}_\Corner^{3d, \A}(H)$ is defined in Definition $2$ of \cite{Hayashi2} and denoted as $\I_\Corner(H)$).
There is the following relation between these two.
\begin{theorem}\label{minusA}
$\mathrm{sf}(\check{\I}_\Corner^{3d, \A}(H)) = - \mathrm{sf}(\hat{\I}_\Corner^{3d, \A}(H))$.
\end{theorem}
\begin{proof}
Let fix a base point of $\T$. We have the isomorphism $K_0(\Sab  {\otimes} C(\T))$ $\cong K_0(\Sab) \oplus K_0(\Sab  {\otimes} C_0((0,1)))$.
The projection onto the second component gives a homomorphism $p \colon K_0(\Sab  {\otimes} C(\T)) \to K_0(\Sab  {\otimes} C_0((0,1)))$.
Let $\theta \colon K_1(\Sab) \to K_0(\Sab  {\otimes} C_0((0,1)))$ be the suspension isomorphism, and let $\beta \colon K_0(K(\HHs)) \to K_1(K(\HHs)  {\otimes} C_0((0,1)))$ be the Bott isomorphism.
Then, by Corollary~\ref{relation}, we have
\begin{gather*}
\mathrm{sf}(\check{\I}_\Corner^{3d, \A}(H)) = \mathrm{sf} \circ \check{\delta}_0 (\I_\BE^{3d, \A}(H))
	= \mathrm{sf} \circ \check{\delta}_0 \circ p (\I_\BE^{3d, \A}(H))\\
	= K_0(\check{\mathrm{Tr}}) \circ \beta^{-1} \circ \check{\delta}_0 \circ p (\I_\BE^{3d, \A}(H))
	= K_0(\check{\mathrm{Tr}}) \circ \check{\delta}_1 \circ \theta^{-1} \circ p (\I_\BE^{3d, \A}(H))\\
	= - K_0(\hat{\mathrm{Tr}}) \circ \hat{\delta}_1 \circ \theta^{-1} \circ p (\I_\BE^{3d, \A}(H))
	= - \mathrm{sf}(\hat{\I}_\Corner^{3d, \A}(H)),
\end{gather*}
where the last equality follows by the repetition of the previous equalities for convex corners.
\end{proof}

We next consider the case of $\alpha = 0$ and $\beta = +\infty$ (we assume $\mu = 0$) and give a construction of an explicit example.
Let $V_3$ be a finite-rank Hermitian vector space.
Let $H_3$ be a multiplication operator on $l^2(\Z^2; V_3)$ generated by a continuous map $\T^2 \to \End(V_3)$.
We assume that $H_3$ is self-adjoint and invertible (Hamiltonian of a 2-D class A  (conventional) topological insulator).
Let $\I^{2d, \mathrm{A}}(H_3)$ be the topological number of $H_3$.
Let $H_2$ be a bounded linear operator $l^2(\Z; V_2)$ introduced in Sect.~$4.1$ (Hamiltonian of a 1-D class AIII (conventional) topological insulator whose chiral symmetry is implemented by $\Pi_2$).
Using these operators, let us consider the following bounded linear self-adjoint operator $H$ on the Hilbert space $l^2(\Z^3; V_3  {\otimes} V_2)$,
\begin{equation}\label{prodHam}
	H = H_3  {\otimes} \Pi_2 + 1  {\otimes} H_2.
\end{equation}
Its partial Fourier transform gives a family of bounded linear self-adjoint operators $\{ H(t) = H_3(t)  {\otimes} \Pi_2 + 1  {\otimes} H_2 \}_{t \in \T}$ on the Hilbert space $l^2(\Z^2; V_3  {\otimes} V_2)$.
\begin{theorem}\label{productA}
We have $\mathrm{sf}(\check{\I}_\Corner^{3d, \A}(H)) = - \I^{2d, \mathrm{A}}(H_3) \cdot \I^{1d, \AIII}(H_2)$, where the right-hand side is the product of two integers.
\end{theorem}
\begin{proof}
By Theorem~$4$ (1) of \cite{Hayashi2}, for our Hamiltonian $H$ of the form (\ref{prodHam}), the edge Hamiltonians $H^0(t)$ and $H^\infty(t)$ are invertible, and thus, the corner invariant is defined.
By Theorem~$4$ (2) of \cite{Hayashi2}, we have $\mathrm{sf}(\hat{\I}_\Corner^{3d, \A}(H)) = \I^{2d, \mathrm{A}}(H_3) \cdot \I^{1d, \AIII}(H_2)$.
Then, the results follow by Theorem~\ref{minusA}.
\end{proof}

By using these results, we provide an explicit example of a bulk Hamiltonian $H$ such that $H^0$ and $H^\infty$ are both gapped and its corner invariant for the concave corner is nontrivial.
\begin{example}\label{examA}
Let $H'_3$ be the following bounded linear self-adjoint operator on $l^2(\Z^2)  {\otimes} \C^2 \cong l^2(\Z^2; \C^2)$:
\vspace{-2mm}
\begin{equation*}
	H'_3 = \frac{1}{2 i} \sum_{j=1,2} (S_j - S_j^*)  {\otimes} \sigma_j + \bigl(-1 + \frac{1}{2} \sum_{j=1,2}(S_j + S_j^*) \bigl)  {\otimes} \sigma_3,
\end{equation*}
where $S_1 = M_{1,0}$ and $S_2 = M_{0,1}$ are translation operators.
$H'_3$ is an example of a $2$-D type A (conventional) topological insulator.
Its topological invariant is calculated in \cite{PS16} and is $\I^{2d, \mathrm{A}}(H'_1) = -1$.
Let $H'_2$ and $\Pi'$ be following self-adjoint operators on the Hilbert space $l^2(\Z)  {\otimes} \C^2 \cong l^2(\Z, \C^2)$:
\begin{equation*}\tiny
	H'_2 = \frac{1}{2} S  {\otimes} (\sigma_1 + i \sigma_2) + \frac{1}{2} S^*  {\otimes} (\sigma_1 - i \sigma_2), \
	\Pi' = 1  {\otimes} \sigma_3.
\end{equation*}
where $\sigma_1$, $\sigma_2$ and $\sigma_3$ are Pauli matrices\footnote{
$
\sigma_1 =
\begin{pmatrix}
0 & 1\\
1 & 0
\end{pmatrix}, \
\sigma_2 =
\begin{pmatrix}
0 & -i\\
i & 0
\end{pmatrix}, \
\sigma_3 =
\begin{pmatrix}
1 & 0\\
0 & -1
\end{pmatrix}.
$}
and $S$ is the translation operator given by $(S \varphi)(n) = \varphi(n-1)$.
Then, we have $\Pi' H'_2 (\Pi')^* = - H'_2$. This is an example of 1-D class AIII (conventional) topological insulator.
Its topological number is $\I^{1d, \AIII}(H'_2) = -1$ (see \cite{PS16}).
By using them, we consider the following bounded linear self-adjoint operator on $l^2(\Z^3;\C^4)$:
\begin{equation*}
	H = H'_3  {\otimes} \Pi' + 1  {\otimes} H'_2.
\end{equation*}
By Theorem~\ref{productA}, its numerical corner invariant for the concave corner is computed as $\mathrm{sf}(\check{\I}_\Corner^{3d,\A}(H)) = - \I^{2d, \mathrm{A}}(H'_1) \cdot \I^{1d,\AIII}(H'_2) = -(-1) \cdot (-1) = -1$.
Note that by Theorem~\ref{minusA}, the numerical corner invariant for convex corner is $\mathrm{sf}(\hat{\I}_\Corner^{3d,\A}(H)) = 1$
(see also Example~$1$ of \cite{Hayashi2}).
\end{example}

\section{Example and 2-D BBH model}\label{sectexam}

In this section, we introduce an explicit example of 2-D class AIII Hamiltonians whose corner invariant is nontrivial on a system with a codimension-two (convex and concave) corner.
Comparing with this example, we discuss Benalcazar--Bernevig--Hughes' 2-D Hamiltonian \cite{BBH17a} from our viewpoint.

We first study the following 1-D class AIII Hamiltonian;
\begin{equation*}
	H_{\AIII}(k; \gamma_1, \gamma_2, \lambda_1, \lambda_2) = H_{\AIII}(k) := \gamma_1 \sigma_1 + \gamma_2 \sigma_2 + \lambda_1 \cos(k) \sigma_1 + \lambda_2 \sin(k) \sigma_2
\end{equation*}
where $k \in \R/2\pi\Z \cong \T$.
Its chiral symmetry is given by $\sigma_3$.
By the Fourier transform, we obtain a bounded linear self-adjoint operator $H_{\AIII}$ on $l^2(\Z, \C^2)$.
For simplicity, we assume $\lambda_1 \neq 0$ and $\lambda_2 \neq 0$ .
Since
\begin{equation*}
H_{\AIII}(k) \hspace{-0.3mm} = \hspace{-0.3mm}
 \begin{pmatrix}
0 & \gamma_1 \hspace{-0.3mm} - \hspace{-0.3mm} i \gamma_2 \hspace{-0.3mm} + \hspace{-0.3mm} \lambda_1 \hspace{-0.3mm} \cos(k) \hspace{-0.3mm} - \hspace{-0.3mm} i \lambda_2 \sin(k) \hspace{-0.3mm} \\
\gamma_1 \hspace{-0.3mm} + \hspace{-0.3mm} i \gamma_2 \hspace{-0.3mm} + \hspace{-0.3mm} \lambda_1 \hspace{-0.3mm} \cos(k) + i \lambda_2 \hspace{-0.3mm} \sin(k) \hspace{-3mm} & 0 \\
\end{pmatrix},
\end{equation*}
the (bulk) Hamiltonian is invertible (i.e. gapped at zero) when $\left| \gamma_1 / \lambda_1\right|^2 + \left| \gamma_2 / \lambda_2\right|^2  \neq 1$.
This is a model of a 1-D class AIII (conventional) topological insulator, and its topological number, which is  the winding number of $\gamma_1 + i \gamma_2 + \lambda_1 \cos(k) + i \lambda_2 \sin(k)$ around zero, is the following.
\vspace{-1mm}
$$
\I^{1d, \AIII}(H_{\AIII}) =
\left\{
\begin{aligned}
1, & \hspace{3mm} \text{if} \ \left| \gamma_1/\lambda_1 \right|^2 + \left| \gamma_2/\lambda_2 \right|^2 < 1,\\
0, & \hspace{3mm} \text{if} \ \left| \gamma_1/\lambda_1 \right|^2 + \left| \gamma_2/\lambda_2 \right|^2 > 1.
\end{aligned}
\right.
$$

Let $\gamma_{x,1}$\hspace{-0.3mm}, $\gamma_{x,2}$\hspace{-0.3mm}, $\gamma_{y,1}$\hspace{-0.3mm}, $\gamma_{y,2}$\hspace{-0.3mm}, $\lambda_{x,1}$\hspace{-0.3mm}, $\lambda_{x,2}$\hspace{-0.3mm}, $\lambda_{y,1}$ and $\lambda_{y,2}$ be real numbers.
By using $H_\AIII$, we consider the following 2-D Hamiltonian,
\vspace{-1mm}
\begin{eqnarray*}
	&H(k_x, k_y; \gamma_{x,1}, \gamma_{x,2}, \lambda_{x,1}, \lambda_{x,2}, \gamma_{y,1}, \gamma_{y,2}, \lambda_{y,1}, \lambda_{y,2}) := \\
		&H_{\AIII}(k_x; \gamma_{x,1}, \gamma_{x,2}, \lambda_{x,1}, \lambda_{x,2}) \otimes 1
				+ \sigma_3 \otimes H_{\AIII}(k_y; \gamma_{y,1}, \gamma_{y,2}, \lambda_{y,1}, \lambda_{y,2})\\
	&= \gamma_{x,1} \sigma_1 \otimes 1 + \gamma_{x,2} \sigma_2 \otimes 1 + \lambda_{x,1} \cos(k_x) \sigma_1 \otimes 1 + \lambda_{x,2} \sin(k_x) \sigma_2 \otimes 1\\
			& + \gamma_{y,1} \sigma_3 \otimes \sigma_1 + \gamma_{y,2} \sigma_3 \otimes \sigma_2 + \lambda_{y,1} \cos(k_y) \sigma_3 \otimes \sigma_1  + \lambda_{y,2} \sin(k_y) \sigma_3 \otimes \sigma_2. 
\end{eqnarray*}
where $k_x, k_y \in \R/2\pi\Z$.
Just for simplicity, we assume that $\lambda_{x,1}$, $\lambda_{x,2}$, $\lambda_{y,1}$ and $\lambda_{y,2}$ are non-zero.
This Hamiltonian preserves the chiral symmetry given by $\Pi = \sigma_3 \otimes \sigma_3$.
Through the Fourier transform, we obtain a bounded linear self-adjoint operator $H$ on $l^2(\Z^2, \C^4)$.
We now take $\alpha = 0$ and $\beta = \infty$ and introduce two edge Hamiltonians $H^0$, $H^\infty$ and the corner Hamiltonian $H^{0,\infty}$.
When $\left| \gamma_{x,1} / \lambda_{x,1}\right|^2 + \left| \gamma_{x,2} / \lambda_{x,2} \right|^2  \neq 1$ and $\left| \gamma_{y,1} / \lambda_{y,1} \right|^2 + \left| \gamma_{y,2} / \lambda_{y,2} \right|^2  \neq 1$, the (bulk) Hamiltonians $H_{\AIII}(k_x)$ and $H_{\AIII}(k_y)$ of 1-D class AIII (conventional) topological insulators are invertible.
Thus, by Sect.~$4.1$ (or Theorem $4$  (1) of \cite{Hayashi2}), the bulk and two edge Hamiltonians ($H$, $H^0$ and $H^\infty$) are invertible and the numerical corner invariant for the convex corner is defined.
Moreover, by Theorem~\ref{prodthmAIII}, its value is the product of topological numbers of two 1-D class AIII (conventional) topological insulators and is computed as follows.
\begin{gather*}
K_0(\hat{\mathrm{Tr}})\hspace{-0.3mm}(\hat{\I}_\Corner^{2d,\AIII}(H)\hspace{-0.3mm}) \hspace{-0.3mm}=\hspace{-0.3mm} K_0(\hat{\mathrm{Tr}})\hspace{-0.3mm}(\hat{\I}_\Corner^{2d,\AIII}(H(\gamma_{x,\hspace{-0.3mm}1}\hspace{-0.3mm},\hspace{-0.3mm} \gamma_{x,\hspace{-0.3mm}2}\hspace{-0.3mm},\hspace{-0.3mm} \lambda_{x,\hspace{-0.3mm}1}\hspace{-0.3mm},\hspace{-0.3mm} \lambda_{x,\hspace{-0.3mm}2}\hspace{-0.3mm},\hspace{-0.3mm} \gamma_{y,\hspace{-0.3mm}1}\hspace{-0.3mm},\hspace{-0.3mm} \gamma_{y,\hspace{-0.3mm}2}\hspace{-0.3mm},\hspace{-0.3mm} \lambda_{y,\hspace{-0.3mm}1}\hspace{-0.3mm}, \hspace{-0.3mm} \lambda_{y,\hspace{-0.3mm}2})\hspace{-0.2mm})\hspace{-0.2mm})\\
 = \I^{1d, \AIII}(H_{\AIII}(k_x; \gamma_{x,1},\hspace{-0.3mm} \gamma_{x,2},\hspace{-0.3mm} \lambda_{x,1},\hspace{-0.3mm} \lambda_{x,2})) \cdot \I^{1d, \AIII}(H_{\AIII}(k_y; \gamma_{y,1},\hspace{-0.3mm} \gamma_{y,2},\hspace{-0.3mm} \lambda_{y,1},\hspace{-0.3mm} \lambda_{y,2}))\\
 = \left\{
\begin{aligned}
1, & \hspace{3mm} \text{if} \ \ \left| \frac{\gamma_{x,1}}{\lambda_{x,1}} \right|^2 + \left| \frac{\gamma_{x,2}}{\lambda_{x,2}} \right|^2 < 1, \ \text{and} \ \left| \frac{\gamma_{y,1}}{\lambda_{y,1}} \right|^2 + \left| \frac{\gamma_{y,2}}{\lambda_{y,2}} \right|^2 < 1\\
0, & \hspace{3mm} \text{otherwise}.
\end{aligned}
\right.\vspace{-1mm}
\end{gather*}
By Theorem~\ref{minusAIII}, the numerical corner invariant $K_0(\check{\mathrm{Tr}})(\check{\I}_\Corner^{2d,\AIII}(H))$ for the concave corner is also (defined and) computed which is their negative.
Thus, when parameters are taken as  $\left| \gamma_{x,1} / \lambda_{x,1}\right|^2 + \left| \gamma_{x,2} / \lambda_{x,2} \right|^2  < 1$ and $\left| \gamma_{y,1} / \lambda_{y,1} \right|^2 + \left| \gamma_{y,2} / \lambda_{y,2} \right|^2  < 1$, there exist topologically protected corner states both for concave and concave corners associated with $\alpha = 0$ and $\beta = \infty$.

\vspace{-1mm}
\begin{remark}\label{shape}
If we change $\alpha$ or $\beta$, the shape/angle of the corner changes.
The previous results \cite{DH71,Pa90,Ji95} and results of Sect.~$2$ and $3$ enables us to treat corners of angles less than $\pi$ and bigger than $\pi$, respectively.
If we fix the bulk Hamiltonian and change $\alpha$ and $\beta$, a natural question is whether numerical corner invariants changes correspondingly.
Example~\ref{examA} and the above one clarify that numerical corner invariants change depending on the shape of the corner.
More precisely, as in Theorem~\ref{minusAIII} and Theorem~\ref{minusA}, numerical corner invariants for concave and convex corners for fixed $\alpha$ and $\beta$ differ by the factor $-1$.
\end{remark}

\begin{remark}\label{BBH}
Let $U$, $r_4$ and $\Theta$ be following transformations on $\C^4$
;\vspace{-1mm}
\begin{equation*}
U :=
 \begin{pmatrix}
0 & 0 & 0 & -1 \\
1 & 0 & 0 & 0 \\
0 & -1 & 0 & 0 \\
0 & 0 & 1 & 0 
\end{pmatrix},
\ \
r_4 :=
 \begin{pmatrix}
0 & 0 & 1 & 0 \\
1 & 0 & 0 & 0 \\
0 & 0 & 0 & -1 \\
0 & 1 & 0 & 0 
\end{pmatrix}.
\ \
\Theta :=
 \begin{pmatrix}
c & 0 & 0 & 0 \\
0 & c & 0 & 0 \\
0 & 0 & c & 0 \\
0 & 0 & 0 & c 
\end{pmatrix}.
\end{equation*}
where\footnote{We here employ the following identification:
\vspace{-1mm}
$
	\begin{pmatrix}
	a & b\\
	c & d
	\end{pmatrix}
	\otimes A
	=
	\begin{pmatrix}
	aA & bA\\
	cA & dA
	\end{pmatrix}.
$} $c$ is the complex conjugation on $\C$.
Matrices $U$ and $r_4$ are unitary transformations and $\Theta$ is an anti-unitary transformation.
If $\gamma_{x,2} = \gamma_{y,2} = 0$ is satisfied, our Hamiltonian preserves two anti-commuting reflection symmetries.
Specifically, let $m_x := -\sigma_1 \otimes \sigma_3$ and $m_y := - 1 \otimes \sigma_1$, then we have,
\begin{equation*}
	m_x H(k_x, k_y)m_x^* = H(-k_x, k_y), \ \ m_y H(k_x, k_y)m_y^* = H(k_x, -k_y).
\end{equation*}
Further, if $\gamma_x = \gamma_y$ and $\gamma_{x,1} = \gamma_{y,1}$ is satisfied, our Hamiltonian preserves time-reversal, particle-hole and $C_4$-symmetries
\begin{gather*}
	\Theta H(k_x, k_y) \Theta^* = H(-k_x, -k_y), \ \
	\Xi H(k_x, k_y) \Xi^* = -H(-k_x, -k_y),\\
	r_4 H(k_x, k_y) r_4^* = H(k_y, -k_x),
\end{gather*}
where $\Xi = \Theta \circ \Pi$.
In other words, we can see that if $\gamma_{x,2} \neq 0$ and $\gamma_{y,2} \neq 0$, two anti-commuting reflection symmetries, the time-reversal symmetry (TRS) and the particle-hole symmetry (PHS) are broken.
If $\lambda_{x,1} \neq \lambda_{y,1}$ or $\lambda_{x,2} \neq \lambda_{y,2}$, the $C_4$-symmetry is broken\footnote{
Note that $r_4 (\sigma_2 \otimes 1) r_4^* = - \sigma_3 \otimes \sigma_2$, $r_4(\sigma_1 \otimes 1)r_4^* = \sigma_3 \otimes \sigma_1$, $ r_4(\sigma_3 \otimes \sigma_1)r_4^* = \sigma_1 \otimes 1$ and $r_4 (\sigma_3 \otimes \sigma_2) r_4^* = \sigma_2 \otimes 1$ holds.
}.

Let us consider the unitary transformation induced by $U$, specifically, consider the following 2-D Hamiltonian\footnote{Note that we have
	$U(\sigma_1 \otimes 1) U^* = \sigma_1 \otimes 1$,
		$U(\sigma_2 \otimes 1) U^* = - \sigma_2 \otimes \sigma_3$,
			$U(\sigma_3 \otimes \sigma_1) U^* = - \sigma_2 \otimes \sigma_2$ and
				$U(\sigma_3 \otimes \sigma_2) U^* = - \sigma_2 \otimes \sigma_1$};
\begin{eqnarray*}
	& U H(k_x, k_y; \gamma_{x,1}, \gamma_{x,2}, \lambda_{x,1}, \lambda_{x,2}, \gamma_{y,1}, \gamma_{y,2}, \lambda_{y,1}, \lambda_{y,2})U^* := \\
	&= \gamma_{x,1} \sigma_1 \otimes 1 - \gamma_{x,2} \sigma_2 \otimes \sigma_3 + \lambda_{x,1} \cos(k_x) \sigma_1 \otimes 1 - \lambda_{x,2} \sin(k_x) \sigma_2 \otimes \sigma_3 \\
			& - \gamma_{y,1} \sigma_2 \otimes \sigma_2 + \gamma_{y,2} \sigma_2 \otimes \sigma_1 + \lambda_{y,1} \cos(k_y) \sigma_2 \otimes \sigma_2  + \lambda_{y,2} \sin(k_y) \sigma_2 \otimes \sigma_1. 
\end{eqnarray*}
When $\gamma_{x,1} = \gamma_{y,1}$, $\gamma_{x,2} = \gamma_{y,2} = 0$ and $\lambda_{x,1} = \lambda_{x,2} = \lambda_{y,1} = \lambda_{y,2}$, this 2-D model is discussed by Benalcazar--Bernevig--Hughes (Equation (6) of \cite{BBH17a}).
In this case, this model preserves TRS, PHS, the chiral symmetry, two anti-commuting reflection symmetries and $C_4$-symmetry specified by the unitary transform of the above operators\footnote{Specifically, they are $U \Theta U^* = \Theta$, $U \Xi U^* = (\sigma_3 \otimes \sigma_1) \circ \Theta$, $U \Pi U^* = \sigma_3 \otimes \sigma_1$, $U m_x U^* = \sigma_1 \otimes \sigma_3$, $U m_y U^* = \sigma_1 \otimes \sigma_1$ and $U r_4 U^* = 
\begin{pmatrix}
	0 & 1 \\
	-i\sigma_2 & 0
	\end{pmatrix}$, respectively.} \cite{BBH17a}.
For this model, they find the quadrupole phase which hosts topologically protected corner states where they stressed the role of reflection symmetries.
Since the unitary transform does not change these topological invariants, as long as we keep track of its chiral symmetry, the above computation also computes the numerical corner invariant of 2-D BBH model both for convex and concave corners associated with $\alpha = 0$ and $\beta = \infty$.
For such a special choice of parameters (as in \cite{BBH17a}), our result about the existence of topologically protected corner states is consistent with that of Benalcazar--Bernevig--Hughes' and gives another explanation for that.
Note that our results states that there exists topologically protected corner states even if we break TRS, PHS, two anti-commuting reflection symmetries and the $C_4$-symmetry.
\end{remark}

\begin{remark}\label{experiment}
After the work of \cite{BBH17a}, corner states are reported to have been observed experimentally in metamaterials \cite{PBHG18,Gracia18}.
\end{remark}

\vspace{-2mm}
\appendix
\section{Some variants}

As in Remark~\ref{othercases}, most results in this paper also hold in the cases in which the corner (or edges) do not necessarily include lattice points on lines $y=\alpha x$ and $y=\beta x$.
In this appendix, we make this statement precise by fixing the setups and clarifying the corresponding results.
Although the proofs of the corresponding results are parallel with those contained in the main body of this paper, some parts of the discussions are based on the explicit construction of an example, especially the constructions of rank-one projections (Lemma~\ref{contain}) and that of the Fredholm concave corner Toeplitz operator of index one (Theorem~\ref{construction}).
For these reasons, we collect the corresponding results in this appendix.
The corresponding results for quarter-plane Toeplitz operators, briefly mentioned in \cite{Ji95}, are also included for completeness.

Since we consider two edges, corresponding to whether the edge includes lattice points on boundaries, we can consider four cases.
Each case corresponds to the case in which closed subspaces $\HHa$ and $\HHb$ of $\HH$ are spanned by the following sets:

\begin{center}
\noindent
Case~$1$ : $\{ {  e_{m,n}} \mid -\alpha m + n \geq 0 \}$ and $\{ {  e_{m,n}} \mid -\beta m + n \leq 0 \}$, respectively.

\noindent
Case~$2$ : $\{ {  e_{m,n}} \mid -\alpha m + n > 0 \}$ and $\{ {  e_{m,n}} \mid -\beta m + n \leq 0 \}$, respectively.

\noindent
Case~$3$ : $\{ {  e_{m,n}} \mid -\alpha m + n \geq 0 \}$ and $\{ {  e_{m,n}} \mid -\beta m + n < 0 \}$, respectively.

\noindent
Case~$4$ : $\{ {  e_{m,n}} \mid -\alpha m + n > 0 \}$ and $\{ {  e_{m,n}} \mid -\beta m + n < 0 \}$, respectively.\
\end{center}

For these cases, we associate concave corners and define concave corner $C^*$-algebras $\TTs$ in the same way as in Sect.~$2$.
Note that Case~$1$ is already treated in the main body of this paper.
In the following, we assume the condition ($\dagger$) for $\alpha$ and $\beta$.

We first collect constructios of rank-one projections in Cases~$2\sim4$.
They correspond to Lemma~\ref{contain} in Case~$1$.
As in Lemma~\ref{contain}, we take $N \in \{ 2,3, \cdots \}$ such that $\frac{1}{N+1} < \alpha \leq \frac{1}{N}$.

\begin{lemma}
In Case~$2 \sim 4$, some $\tcP_k$ is a rank-one projection.
Explicitly, we have the following results.

\noindent
In Case~$2$,
$\begin{cases}
	\text{when} \ \frac{1}{N+1} < \alpha \leq \frac{1}{N} \ \text{and} \ \beta = 1, \text{we have} \ \tcP_{N-1} = p_{-N-1,-1}.\\
	\text{when} \ \frac{1}{N+1} < \alpha \leq \frac{1}{N} \ \text{and} \ 1 < \beta < \infty, \text{we have} \ \tcP_N = p_{-N-1,-1}.
\end{cases}$

\noindent
In Case~$3$,
$\begin{cases}
	\text{when} \ \alpha = \frac{1}{N} \ \text{and} \ 1 < \beta \leq \infty, \ \text{we have} \ \tcP_{N-1} = p_{-N,-1}.\\
	\text{when} \ \frac{1}{N+1} < \alpha < \frac{1}{N} \ \text{and} \ 1 < \beta \leq \infty, \text{we have} \ \tcP_N = p_{-N-1,-1}.
\end{cases}$

\noindent
In Case~$4$,
$\begin{cases}
\text{when} \ \alpha = \frac{1}{N} \ \text{and} \ \beta = 1, \ \text{we have} \ \tcP_{1} = p_{-1,0}.\\
\text{in the other cases (under $(\dagger)$), we have} \ \tcP_{N} = p_{-N-1,-1}.
\end{cases}$
\end{lemma}
\vspace{-2mm}

We here write down the result of computing the Fredholm index of the following operator in Cases $2\sim4$ which corresponds to Theorem~\ref{construction} in Case~$1$.
\vspace{-1mm}
\begin{equation*}
	\check{A} := \check{\cP}_{0,1} + M_{1,1}(1 - \check{\cP}_{-1,0}) + M_{1,0}(\check{\cP}_{-1,0} - \check{\cP}_{0,1}).
\end{equation*}
\vspace{-6mm}
\begin{proposition}
In Cases~$2\sim4$, $\check{A}$ is a surjective Fredholm operator whose Fredholm index is $1$.
We also have $\check{A} - 1 \in \cCs$.
Its kernel is given as follows:

\noindent
In Case~$2$,
$\begin{cases}
\text{when} \ 0 < \alpha \leq \frac{1}{2} \ \text{and} \ \beta = 1, \Ker \check{A} = \C (\e_{-2, 0} - \e_{-1,0}).\\
\text{when} \ 0 < \alpha \leq \frac{1}{2} \ \text{and} \ 1 < \beta < \infty, \Ker \check{A} = \C (\e_{-1,0} - \e_{0,0}).
\end{cases}$

\noindent
In Case~$3$, under the assumption $(\dagger)$, we have $\Ker \check{A} = \C (\e_{-1,0} - \e_{0,0})$.

\noindent
In Case~$4$, under the assumption $(\dagger)$, we have $\Ker \check{A} = \C (\e_{0,1} - \e_{1,1})$.
\end{proposition}

We next consider the following quarter-plane Toeplitz operator in Cases~$1\sim4$.
\begin{equation*}
	\hat{A} := \hat{\cP}_{0,1} + M_{1,1}(1 - \hat{\cP}_{-1,0}) + M_{1,0}(\hat{\cP}_{-1,0} - \hat{\cP}_{0,1}).
\end{equation*}
Note that $\hat{A} \in \TTab$.
Jiang shows in \cite{Ji95} that, under the assumption ($\dagger$), this is an isometric Fredholm operator and compute its Fredholm index mainly in the Case~$1$.
The other cases are briefly mentioned (Remark (1) in p2828 of \cite{Ji95}), though their Fredholm indices are stated as $\pm 1$.
We here need to fix its sign in order to obtain the corresponding result for Corollary~\ref{relation} especially in Cases~$2\sim4$.
For this reason, we (re)state necessary results in the following.
Its proof is totally parallel with that of Jiang \cite{Ji95}.

\begin{proposition}[Jiang \cite{Ji95}]
In Cases~$1\sim4$, $\hat{A}$ is an isometric Fredholm operator whose Fredholm index is $-1$.
Its cokernel is given as follows:

\noindent
In Case~$1$, under the assumption $(\dagger)$, we have $\Coker \hat{A} = \C \e_{0,0}$.

\noindent
In Case~$2$, under the assumption $(\dagger)$, we have $\Coker \hat{A} = \C \e_{1,1}$.

\noindent
In Case~$3$,
$\begin{cases}
\text{when} \ 0 < \alpha \leq \frac{1}{2} \ \text{and} \ \beta = 1, \Coker \hat{A} = \C \e_{2,1}.\\
\text{when} \ 0 < \alpha \leq \frac{1}{2} \ \text{and} \ 1 < \beta < \infty, \Coker \hat{A} = \C \e_{1,1}.
\end{cases}$

\noindent
In Case~$4$,
$\begin{cases}
\text{when} \ \alpha = \frac{1}{2} \ \text{and} \ \beta = 1, \Coker \hat{A} = \C \e_{3,2}.\\
\text{when} \ 0 < \alpha < \frac{1}{2} \ \text{and} \ \beta = 1, \Coker \hat{A} = \C \e_{2,1}.\\
\text{when} \ 0 < \alpha \leq \frac{1}{2} \ \text{and} \ 1 < \beta < \infty, \Coker \hat{A} = \C \e_{1,1}.
\end{cases}$
\end{proposition}

\subsection*{Acknowledgments}
The author would like to thank Takeshi Nakanishi and Yukinori Yoshimura for showing him the result of a numerical calculation, which convinced him about the content of this paper.
He also would like to thank Ken-Ichiro Imura  and Ryo Okugawa for many discussions concerning \cite{BBH17a} and Max Lein for sharing the information regarding \cite{Gracia18}.
The author acknowledge the support of the Erwin Schr{\"o}dinger Institute where part of this work was conducted.
He would like to thank organizers of the workshop ``Bivariant K-theory in Geometry and Physics'' for their hospitability.
This work was supported by JSPS KAKENHI Grant Number JP17H06461 and JP19K14545.


\end{document}